\documentclass[onecolumn,accepted=2021-02-22,a4paper]{quantumarticle}
\pdfoutput=1

\usepackage{geometry}
\usepackage[numbers,sort&compress]{natbib}
\usepackage{graphicx}
\usepackage{dcolumn}
\usepackage{dsfont}
\usepackage{enumitem}
\usepackage{bm}
\usepackage{amsmath,amssymb,amsthm}
\usepackage{braket}
\usepackage{cancel}

\usepackage{changepage}

\usepackage[
	colorlinks=true,
	urlcolor = black,
	linkcolor    = blue!60!black, %Colour of internal links
	citecolor=green!60!black
	]{hyperref}

\usepackage{algpseudocode,algorithm,algorithmicx}

\usepackage{tikz}

\newcommand\encircle[1]{%
  \tikz[baseline=(X.base)] 
    \node (X) [draw, shape=circle, inner sep=0] {\strut #1};}

\usepackage{xcolor}
\usepackage{url}
\usepackage{microtype}
\usepackage{soul}

\newcommand{\gen}{\textsf{GEN}}
\newcommand{\kgen}{\textsf{KGEN}}
\newcommand{\qgen}{\textsf{QGEN}}
\newcommand{\query}{\textsf{Query}}

\newcommand{\md}{\text{ mod}}
\newcommand{\mq}{\mathrm{MQ}}
\newcommand{\pex}{\mathrm{PEX}}
\newcommand{\qmq}{\mathrm{QMQ}}

\newcommand{\qpex}{\mathrm{QPEX}}
\newcommand{\rpex}{\mathrm{RPEX}}

\newcommand{\sample}{\mathrm{SAMPLE}}
\newcommand{\qsample}{\mathrm{QSAMPLE}}
\newcommand{\bin}{\mathrm{BIN}}

\newtheorem{theorem}{Theorem}

\newtheorem{question}{Question}

\newtheorem{corollary}{Corollary}[theorem]
\newtheorem{lemma}{Lemma}

\newtheorem{observation}{Observation}
\newtheorem{conjecture}{Conjecture}
\newtheorem{definition}{Definition}
\newtheorem{claim}{Claim}

\begin{document}

\author{Ryan Sweke}
\affiliation{\mbox{Dahlem Center for Complex Quantum Systems, Freie Universit\"{a}t Berlin, D-14195 Berlin, Germany}}

\author{Jean-Pierre Seifert}
\affiliation{\mbox{Department of Electrical Engineering and Computer Science,
TU Berlin, D-10587 Berlin, Germany}}
\affiliation{\mbox{FhG SIT, D-64295 Darmstadt, Germany}}

\author{Dominik Hangleiter}
\affiliation{\mbox{Dahlem Center for Complex Quantum Systems, Freie Universit\"{a}t Berlin, D-14195 Berlin, Germany}}

\author{Jens Eisert}
\affiliation{\mbox{Dahlem Center for Complex Quantum Systems, Freie Universit\"{a}t Berlin, D-14195 Berlin, Germany}}

\affiliation{\mbox{Helmholtz Center Berlin, D-14109 Berlin, Germany}}
\affiliation{\mbox{Department of Mathematics and Computer Science, Freie Universit{\"a}t Berlin, D-14195 Berlin, Germany}}

% \date{\today}

\title{On the Quantum versus Classical Learnability of Discrete Distributions}

\begin{abstract}
Here we study the comparative power of classical and quantum learners for generative modelling within the Probably Approximately Correct (PAC) framework. More specifically we consider the following task: Given samples from some unknown discrete probability distribution, output with high probability an efficient algorithm for generating new samples from a good approximation of the original distribution. Our primary result is the explicit construction of a class of discrete probability distributions which, under the decisional Diffie-Hellman assumption, is provably not efficiently PAC learnable by a classical generative modelling algorithm, but for which we construct an efficient quantum learner. This class of distributions therefore provides a concrete example of a generative modelling problem for which quantum learners exhibit a provable advantage over classical learning algorithms. In addition, we discuss techniques for proving classical generative modelling hardness results, as well as the relationship between the PAC learnability of Boolean functions and the PAC learnability of discrete probability distributions.
\end{abstract}

\maketitle

\section{Introduction}\label{s:intro}
Since its introduction, Valiant's model of ``Probably Approximately Correct" (PAC) learning \cite{valiant1984theory}, along with a variety of natural extensions and modifications, has provided a fruitful framework for studying both the computational and statistical aspects of machine learning \cite{kearns1994introduction,shalev2014understanding}. Importantly, the PAC framework also provides a natural setting for the rigorous comparison of quantum and classical learning algorithms \cite{arunachalam2017survey,ciliberto2020statistical}. In fact,  while the recent availability of ``noisy intermediate scale quantum" (NISQ) devices has spurred a huge interest in the potential of quantum enhanced learning algorithms \cite{dunjko2018machine, schuld2018supervised, biamonte2017quantum}, it is interesting to note that there is a rich history of quantum learning theory, beginning as early as 1995 with the seminal work of Bshouty and Jackson \cite{arunachalam2017survey, bshouty1998learning}. Despite this rich history, the majority of previous work on quantum learning theory has focused on the classical versus quantum learnability of different classes of Boolean functions, which provides an abstraction of supervised learning \cite{shalev2014understanding}. 

In this work, we study the classical versus quantum PAC learnability of discrete probability distributions. More specifically, at an informal level we explore the following question, from the perspective of both classical and quantum learning algorithms: Given samples from some unknown probability distribution, output with high probability an efficient algorithm for generating new samples from a good approximation of the original distribution. 
We refer to this task as \textit{generative modelling}. Note that one could also consider the related problem not of generating new samples from a distribution, but of learning a description of the distribution itself -- a problem known as \textit{density estimation} \cite{canonne2020short, kamath2015learning, diakonikolas2016learning}. 

Here, we focus exclusively on generative modelling for a variety of reasons. Firstly, from a purely classical perspective, modern heuristic models and algorithms for generative modelling, such as Generative Adversarial Networks (GANS) \cite{goodfellow2014generative}, Variational Auto-Encoders~\cite{kingma2013auto} and Normalizing Flows \cite{Kobyzev_2020} have proven extremely successful, with a wide variety of practical applications, and as such understanding whether quantum algorithms may be able to offer an advantage for this task is of natural interest. Additionally, a variety of quantum models and algorithms for generative modelling have recently been proposed, such as Born Machines~\cite{coyle2019born, liu2018differentiable, Benedetti_2019, Gaoeaat9004}, Quantum GANS~\cite{Dallaire_Demers_2018,Hueaav2761, lloyd2018quantum, chakrabarti2019quantum} and Quantum Hamiltonian-based models~\cite{verdon2019quantum}. While the majority of these approaches remain ill-understood from a theoretical perspective, Ref. \cite{Gaoeaat9004} has indeed already established evidence for a meaningful generative modelling advantage, in some specific instances, when using a particular tensor network inspired quantum generative model. Furthermore, we know that there exist probability distributions which cannot be efficiently sampled from classically, but which can be efficiently sampled from using quantum devices \cite{BosonSampling, bremner_average-case_2016, arute2019quantum}. In light of this fact, and the emergence of quantum algorithms for generative modelling, Ref.~\cite{coyle2019born} has formalized the question, within the PAC framework, of whether there also exist classes of probability distributions which can be efficiently \textit{PAC learned} (in a generative sense) with quantum resources, but not with purely classical approaches. Our primary contribution in this work is to answer this question in the affirmative, through the explicit construction of a concept class of discrete probability distributions which, under the Decisional Diffie-Hellman (DDH) assumption~\cite{boneh1998decision}, is provably not efficiently PAC learnable from classical samples by a classical learning algorithm, but for which we provide an efficient quantum PAC learning algorithm. This class of distributions therefore provides a concrete example of a generative modelling problem for which quantum learners exhibit a provable advantage over classical learning algorithms, within the PAC framework.

The following important points regarding the setting of the result are worth clarifying. Firstly, although it might seem natural to consider the learnability of probability distributions describing the outcome of quantum processes, such as the measurement of a parameterized quantum state or random quantum circuit, we focus exclusively in this work on probability distributions describing the outcomes of \textit{classical} circuits. Apart from allowing us to exploit existing results concerning the hardness of learnability for specific classes of discrete probability distributions \cite{Kearns:1994:LDD:195058.195155}, this restriction also allows us to demonstrate a quantum generative modelling advantage for purely classical problems. Additionally, in the context of Boolean function learning, it is often of interest to consider quantum learning algorithms which have access to \textit{quantum examples} -- in essence a superposition of input/output pairs from the unknown function to be learned \cite{arunachalam2017survey}. In the setting we are concerned with here, it is also possible to consider a notion of \textit{quantum samples} from a distribution, however once again we choose to restrict ourselves to quantum learning algorithms with access to classical samples from the unknown probability distribution, which provides the fairest comparison of quantum versus classical learners for generative modelling. Finally, it is important to stress that the efficient quantum learning algorithm which we provide is expected to require a universal fault-tolerant quantum computer, as it makes use of the exact efficient quantum algorithm for discrete logarithms \cite{mosca2004exact}. As a result, we \emph{do not} expect that the separation we show in this work can be experimentally demonstrated on NISQ devices. Studying the learnability of probability distributions generated by quantum processes, the power of learners with quantum samples, and the power of near-term quantum learning algorithms remain interesting open problems, and as such we will also discuss the consequences of our results and techniques for approaching these questions. 

To provide a generative modelling task for which there exists a definitive provable separation between the power of quantum and classical learners, we rely heavily on techniques at the rich interface of computational learning theory and cryptography \cite{kearns1994introduction}. More specifically, we start from the prior work of Kearns, Mansour, Ron, Rubinfeld, Schapire and Sellie (KMRRSS) \cite{Kearns:1994:LDD:195058.195155}, who have shown that given any pseudorandom function collection it is possible to construct a class of probability distributions for which no efficient classical generative modelling algorithm exists. We show that in order for such a class of distributions to be efficiently quantum learnable, one requires a pseudorandom function collection for which there exists a quantum adversary, who in addition to distinguishing keyed instances of the pseudorandom function collection from random functions via membership queries, can also learn the secret key using only random examples. By using the DDH assumption as a primitive, we are then able to construct such a pseudorandom function collection via a slight modification of the Goldreich-Goldwasser-Micali (GGM) construction~\cite{Goldreich:1986:CRF:6490.6503}.

Although the classical hardness result of KMRRSS \cite{Kearns:1994:LDD:195058.195155} is a sufficient starting point for our purposes, we also address in this work the possibility of obtaining similar classical hardness results for generative modelling from primitives other than pseudorandom function collections. More specifically, we formulate and discuss conjectures concerning the possibility of proving classical hardness results for generative modelling from both \textit{weak} pseudorandom function collections, and from existing hardness results for the PAC learnability of \textit{Boolean functions}. Apart from being of conceptual interest, in the former case these considerations are motivated by the possibility of using such results to address questions concerning generative modelling with near-term quantum learners, as well as quantum learners with quantum samples. In the latter case, these considerations are motivated by a desire to understand better the relationship between the PAC learnability of discrete probability distributions, and the PAC learnability of Boolean functions.

From the above outline one can see that both the results and techniques of this work lie at the intersection of quantum machine learning, computational learning theory and cryptography. In particular, while our primary result is very much in the spirit of computational learning theory, and contributes new ideas and techniques in this vein, it is also certainly of interest to the quantum machine learning community, and largely motivated by a desire to understand more clearly the potential and limitations of quantum enhanced machine learning. As a result, in order for this work to be accessible to readers with differing backgrounds and interests, we will provide a detailed and pedagogical presentation of the foundational material necessary for understanding both the context and details of our main result.

We proceed in this work as follows: Firstly, we begin in Section \ref{s:PAC} with an introduction to the PAC framework, both for concept classes consisting of Boolean functions, and for the generative modelling of concept classes consisting of probability distributions over discrete domains. Given these foundations, we conclude Section \ref{s:PAC} with the statement of Question~\ref{q:sample_vs_sample}, which provides a precise technical description of the primary question that we address in this work, and which we have described informally above. With this in hand, we then proceed in Section~\ref{s:Seperation} to answer Question~\ref{q:sample_vs_sample} in the affirmative. More specifically, after providing an overview of the necessary cryptographic notions in Section~\ref{ss:crypto}, we present in Section~\ref{ss:kearns} a technique due to KMRRSS\ \cite{Kearns:1994:LDD:195058.195155} for constructing from any pseudorandom function collection a distribution class which is provably not efficiently learnable by classical learning algorithms. This technique then allows us to construct in Section~\ref{ss:DDH_result} a distribution class which, under the DDH assumption, is provably not efficiently learnable by any classical learning algorithm, but for which we provide explicitly an efficient quantum learner for the generative modelling task. We then briefly discuss in Section \ref{ss:verification} a method for the \textit{verification} of the advantage exhibited by the quantum learner we provide. Having fully addressed Question \ref{q:sample_vs_sample} at this point, we then shift gears and explore in Section \ref{s:alt_classical_hardness} the possibility of obtaining classical generative modelling hardness results from primitives other than pseudorandom function collections. In particular, in Section \ref{ss:weak_PRFs} we discuss whether weak pseudorandom function collections would be sufficient, and in Section \ref{ss:from_bool_to_dist} we examine the relationship between PAC learnability of Boolean functions, and the PAC generative modelling of associated probability distributions. Finally, in Section \ref{s:conclusion} we summarize our results, and provide an overview of interesting related and open questions, focusing specifically on the setting of probability distributions generated by quantum processes.

\section{Quantum and Classical PAC Learning}\label{s:PAC}

In this section, we begin by defining the notion of \textit{probably approximately correct} (PAC) learnability, both for concept classes consisting of Boolean functions, and concept classes consisting of probability distributions over discrete domains. As we will see, these notions provide a meaningful abstract framework for studying computational aspects of both supervised learning and probabilistic/generative modelling respectively. While the main result of this work is concerned with the latter setting, we begin with the more familiar context of Boolean functions in order to introduce both the fundamental ideas, and a variety of oracle models which will be important throughout this work. Additionally, as mentioned in the introduction, after presentation of our main distribution learning results in Section \ref{s:Seperation}, we will in Section \ref{ss:from_bool_to_dist} discuss in detail the relationship between PAC learnability of Boolean function classes, and PAC learnability of discrete distribution classes.   

\subsection{PAC Learning of Boolean Functions}\label{ss:PAC_Boolean}

Let us denote by $\mathcal{F}_n$ the set of all Boolean functions on $n$ bits -- i.e. $\mathcal{F}_n = \{f|f:\{0,1\}^n \rightarrow \{0,1\}\}$. Notice that any function in $\mathcal{F}_n$ can be specified via its truth table, and therefore $\mathcal{F}_n \simeq \{0,1\}^{2^n}$. We call any subset $\mathcal{C} \subseteq \mathcal{F}_n$ a \emph{concept class}. For any $f\in \mathcal{F}_n$ we can define various types of classical and quantum oracle access to $f$. Classically, we define the \emph{membership query} oracle $\mq(f)$ as the oracle which on input $x$ returns $f(x)$, and the \emph{random example} oracle $\pex(f,D)$ as the oracle which when invoked returns a tuple $(x,f(x))$, where $x$ is drawn from the distribution $D$ over $\{0,1\}^n$. It will also be useful to us later to define the oracle $\rpex(f,D)$ which when invoked returns only $f(x)$, with $x$ drawn from $D$. This can be summarized as follows: 
\begin{align}
    &\query[\mq(f)](x) = f(x),\\
    &\query[\pex(f,D)] = (x,f(x)) \text{ with } x \leftarrow D, \\
    &\query[\rpex(f,D)] = f(x) \text{ with } x \leftarrow D,
\end{align}
where we have used the notation $x\leftarrow D$ to indicate that $x$ is drawn from $D$. Additionally, we define the \emph{quantum membership query} oracle $\qmq(f)$ as the oracle which on input $|x\rangle\otimes|y\rangle$ returns $|x\rangle\otimes|f(x)\oplus y\rangle$, and the \emph{quantum random example} oracle $\qpex(f,D)$ as the oracle which when invoked returns the quantum state $\sum_x\sqrt{D(x)}|f(x)\rangle$, where again $D$ is some distribution over $\{0,1\}^n$. This can be summarized as
\begin{align}
    &\query[\qmq(f)](|x\rangle\otimes|y\rangle) = |x\rangle|\otimes |f(x)\oplus y\rangle,\\
    &\query[\qpex(f,D)] = \sum_{x\in\{0,1\}^n}\sqrt{D(x)}|f(x)\rangle.
\end{align}
As it will be convenient later, we also define $\mq(f,D) := \mq(f)$ and $\qmq(f,D) := \qmq(f)$ for all distributions $D$. 
For a more detailed discussion of these oracles, and in particular the motivation behind their definitions and the relationships between them, we refer to Ref.\ \cite{arunachalam2017survey}. Given these notions, we can then formulate the following definition of a PAC learner for a given concept class:

\begin{definition}[PAC Learners]\label{d:pac_learners}
An algorithm $\mathcal{A}$ is an $(\epsilon,\delta,O, D)$-PAC learner for a concept class $\mathcal{C} \subseteq \mathcal{F}_n$, if for all $c\in \mathcal{C}$, when given access to oracle $O(c,D)$, with probability at least $1-\delta$, the learner $\mathcal{A}$ outputs a hypothesis $h \in \mathcal{F}_n$ such that
\begin{equation}
\underset{x\leftarrow D}{\mathrm{Pr}}[h(x) \neq c(x)] \leq \epsilon.
\end{equation}
An algorithm $\mathcal{A}$ is an $(\epsilon,\delta,O)$-PAC learner for a concept class $\mathcal{C}$, if it is an $(\epsilon,\delta,O, D)$-PAC learner for all distributions $D$.
\end{definition}
\noindent Before continuing, it is useful to make some comments concerning this definition. Firstly, note that the above formulation allows us to consider both classical learning algorithms, with either classical membership query or classical random example oracle access, as well as quantum learning algorithms, with either classical or quantum oracle access of any type. Additionally, it is important to point out that for a given model of oracle access $O$ we can consider either distribution dependent learners -- i.e. learners which are required to succeed (in the sense of being probably approximately correct) only with respect to samples drawn from some fixed distribution $D$, or distribution independent learners, which should succeed with respect to samples drawn from all possible distributions. In light of these observations, we see that Definition \ref{d:pac_learners} provides for us a flexible abstraction of supervised learning, which allows for the comparison of a variety of different learning algorithms, each of which models the supervised learning problem in a different context. In order to perform a meaningful \textit{computational} comparison of these different learning algorithms, we need the following notions of sample and time complexity.

\begin{definition}[Complexity of PAC Learners]\label{d:pac_complexity}
The sample (time) complexity of an $(\epsilon,\delta,O, D)$-PAC learner $\mathcal{A}$ for a concept class $\mathcal{C}$ is the maximum number of queries made by $\mathcal{A}$ to the oracle $O(c,D)$ (maximum run-time required by $\mathcal{A}$) over all $c\in\mathcal{C}$, and over all internal randomness of the learner. The sample (time) complexity of an $(\epsilon,\delta, O)$-PAC learner $\mathcal{A}$ for a concept class $\mathcal{C}$ is the maximum number of queries made by $\mathcal{A}$ to the oracle $O(c,D)$ (maximum run-time required by $\mathcal{A}$) over all $c\in\mathcal{C}$, all possible distributions $D$ and all internal randomness of the learner. Both an $(\epsilon,\delta,O,D)$-PAC and an  $(\epsilon,\delta,O)$-PAC learner for a concept class $\mathcal{C} \subseteq \mathcal{F}_n$ is called efficient if its time complexity is $\mathcal{O}(\mathrm{poly}(n,1/\delta,1/\epsilon))$. 
\end{definition}
\noindent Given this, the following definition formalizes a variety of notions for the efficient PAC learnability of a concept class:
\begin{definition}[Efficient PAC Learnability of a Concept Class]\label{d:pac_concept}
We say that a concept class $\mathcal{C}$ is efficiently classically (quantum) PAC learnable with respect to distribution $D$ and oracle $O$ if for all $ 0< \epsilon,\delta < 1$ there exists an efficient classical (quantum) $(\epsilon,\delta,O,D)$-PAC learner for $\mathcal{C}$. Similarly, $\mathcal{C}$ is efficiently classically (quantum) PAC learnable with respect to oracle $O$ if for all $ 0< \epsilon,\delta < 1$ there exists an efficient classical (quantum) $(\epsilon,\delta,O)$-PAC learner for $\mathcal{C}$.
\end{definition}
\noindent For a complete overview of known results and open questions concerning classical versus quantum learnability of Boolean function concept classes, we again refer to Ref.\ \cite{arunachalam2017survey}. 

\subsection{PAC Learning of Discrete Distributions}\label{ss:PAC_Dist}

\noindent In the previous section we provided definitions for the PAC learnability of concept classes consisting of Boolean functions, which provides an abstract framework for studying and comparing computational properties of different supervised learning algorithms. In this section, we formulate a generalization to concept classes consisting of discrete distributions, which builds on and refines the prior work of Refs. \cite{coyle2019born, Kearns:1994:LDD:195058.195155}, and provides an abstract framework for studying probabilistic modelling from a computational perspective. Additionally, this formulation allows us to state precisely the primary question that we address in this work. For simplicity (and without loss of generality) we will consider distributions over bit strings, and as such we denote the set of all distributions over $\{0,1\}^n$ as $\mathcal{D}_n$, and we call any $\mathcal{C} \subseteq \mathcal{D}_n$ a \emph{distribution class}. We also denote the uniform distribution over $\{0,1\}^n$ as $U_n$. In order to provide a meaningful generalization of PAC learning to this setting, the first thing that we require is a meaningful notion of a query to a distribution. To do this, given some distribution $D \in \mathcal{D}_n$, we define the \emph{sample} oracle $\sample(D)$ as the oracle which when invoked returns some $x$ drawn from $D$. More specifically, we have that 
\begin{equation}
    \query[\sample(D)] = x \leftarrow D.
\end{equation}
Additionally, it is natural to define the \emph{quantum sample} oracle $\qsample(D)$ via
\begin{equation}\label{e:qsample}
    \query[\qsample(D)] = \sum_{x\in \{0,1\}^n}\sqrt{D(x)}|x\rangle.
\end{equation}
In particular, note that given access to $\qsample(D)$ one can straightforwardly simulate access to $\sample(D)$ by simply querying $\qsample(D)$, and then performing a measurement in the computational basis. At this point, it is important to point out that unlike in the case of function concept classes -- where what it means to ``learn a function" is relatively straightforward -- there are two distinct notions of what it means to ``learn a distribution" \cite{Kearns:1994:LDD:195058.195155}. Informally, given some unknown distribution $D$, as well as access to either a classical or quantum sample oracle, we could ask that a learning algorithm outputs an \emph{evaluator} for $D$ -- i.e. some function $\tilde{D}:\{0,1\}^n \rightarrow [0,1]$ which on input $x \in \{0,1\}^n$ outputs an estimate for $D(x)$, and therefore provides an approximate description of the distribution. This is perhaps the most intuitive notion of what it means to learn a probability distribution, and one can indeed construct a corresponding notion of PAC learnability \cite{Kearns:1994:LDD:195058.195155}, for which a variety of results are known for different distribution classes~\cite{Kearns:1994:LDD:195058.195155,canonne2020short,kamath2015learning,diakonikolas2016learning}. However, in many practical settings one might not be interested in learning a full description of the probability distribution (an evaluator for the probability of events) but rather in being able to generate samples from the distribution. As such, instead of asking for a description of the unknown probability distribution (an evaluator) we could ask that the learning algorithm outputs a \emph{generator} for $D$ -- i.e. a probabilistic (quantum or classical) algorithm which when run generates samples from $D$. From a heuristic perspective we note that many of the most widely utilized probabilistic modelling architectures and algorithms, such as generative adversarial networks, are precisely learning algorithms of this type. Additionally, there has recently been a surge of interest in \emph{quantum} learning algorithms of this type -- so called \textit{Born machines}~\cite{coyle2019born, liu2018differentiable, Benedetti_2019} -- which are based on the simple observation that one can sample from a given distribution by preparing and measuring an appropriate quantum state (such as the state provided by the $\qsample$ oracle). We note that interestingly ``learning to evaluate" and ``learning to generate" are incomparable learning problems, in the sense that being able to learn an evaluator does not imply being able to learn a generator and vice versa~\cite{Kearns:1994:LDD:195058.195155}. While the learning of evaluators is certainly both interesting and important, with many open questions \cite{diakonikolas2016learning}, in this work we will focus exclusively on the problem of learning generators for distribution classes. To this end, we start with the following definition, adapted from Refs.~\cite{coyle2019born,Kearns:1994:LDD:195058.195155}, which formalizes the notions of efficient classical and quantum generators: 

\begin{definition}[Efficient Classical and Quantum Generators]\label{d:eff_gen}  Given some probability distribution $D$ over $\{0,1\}^n$, we say that a classical algorithm $\gen_D$ (or quantum algorithm $\qgen_D$) is an efficient classical (quantum) generator for $D$ if $\gen_D$ ($\qgen_D$) produces samples in $\{0,1\}^n$ according to $D$, using $\mathrm{poly}(n)$ resources. In the case of a classical algorithm, we allow the generator to receive as input $m=\mathrm{poly}(n)$ uniformly random input bits.
\end{definition}

\noindent We say that a distribution class $\mathcal{C}$ can be efficiently classically (quantum) generated if for all $D\in \mathcal{C}$ there exists an efficient classical (quantum) generator for $D$. Additionally, again following Refs. \cite{coyle2019born, Kearns:1994:LDD:195058.195155}, we can define the following notion of an approximate generator, which is necessary for a meaningful notion of PAC learnability:

\begin{definition}[Classical and Quantum $(d,\epsilon)$-Generator]\label{d:approx_generator}
Let $d$ be some distance measure on the space of probability distributions over $\{0,1\}^n$, and $D$ some probability distribution over $\{0,1\}^n$. Given some other distribution $D'$ over $\{0,1\}^n$, as well as an efficient classical generator $\gen_{D'}$ (or quantum generator $\qgen_{D'}$) for $D'$, we say that $\gen_{D'}$ ($\qgen_{D'}$) is an efficient classical (quantum) $(d,\epsilon)$-generator for $D$ if $d(D,D') \leq \epsilon$. 
\end{definition}
\noindent In this work we will use primarily the Kullback-Leibler (KL) divergence, defined via
\begin{equation}
    d_{\mathrm{KL}}(D,D') := \sum_{x}D(x)\log\left(\frac{D(x)}{D'(x)}\right),
\end{equation}
as well as the total variation (TV) distance
\begin{equation}
    d_{\mathrm{TV}}(D,D') := \frac{1}{2}\sum_{x}|D(x) - D'(x)|.
\end{equation}
We note that by virtue of its asymmetry the KL-divergence is not strictly a metric, however, via Pinsker's inequality we have that
\begin{equation}\label{e:pinsker}
    d_{\mathrm{TV}}(D,D') \leq \left(\sqrt{\frac{\ln(2)}{2}}\right)\sqrt{d_{\mathrm{KL}}(D,D')}.
\end{equation}
For more details on the interpretation of these and other relevant distance measures, we refer to Ref.~\cite{lehmann2006testing}. Given these preliminaries the following definition provides a natural generalisation of Definition~\ref{d:pac_learners} to the generative modelling context in which we are interested:

\begin{definition}[PAC Generator Learners]\label{d:pac_gen_learner}  A learning algorithm $\mathcal{A}$ is an $(\epsilon,\delta, O, d)$-PAC $\gen$-learner ($\qgen$-learner) for a distribution class $\mathcal{C}$, if for all $D \in \mathcal{C}$, when given access to oracle $O(D)$, with probability $1-\delta$ the learner $\mathcal{A}$ outputs a classical $(d,\epsilon)$-generator $\gen_{D'}$ (quantum $(d,\epsilon)$-generator $\qgen_{D'}$) for $D$.
\end{definition}

\noindent Before continuing it is worth clarifying two important aspects of Definition \ref{d:pac_gen_learner} (also illustrated in Figure~\ref{f:gen_pac_fig}):
\begin{enumerate}
    \item The learning algorithm $\mathcal{A}$ could be either classical or quantum, and in the quantum case the learner could have access to either the classical or quantum sample oracle (i.e. in the quantum case $O$ could be either $\sample$ or $\qsample$).
    \item Both classical and quantum learning algorithms could output either a classical generator or a quantum generator. In the former case we refer to the learner as a $\gen$-learner, and in the latter case as a $\qgen$-learner. In particular, while perhaps counterintuitive, we could consider classical $\qgen$-learners (which could for example output a description of a quantum sampling circuit) as well as quantum $\gen$-learners (which could output descriptions of classical circuits).
\end{enumerate}
\noindent Given the above definition, we can now define the sample/time complexity of  PAC generator learners analogously to how we have defined these notions in Definition \ref{d:pac_complexity}:

\begin{definition}[Complexity of PAC Generator Learners]\label{d:pac_gen_complexity}
The sample (time) complexity of either an $(\epsilon,\delta,O, d)$-PAC $\gen$-learner or an $(\epsilon,\delta,O, d)$-PAC $\qgen$-learner $\mathcal{A}$ for a distribution class $\mathcal{C}$ is the maximum number of queries made by $\mathcal{A}$ to the oracle $O(D)$ (maximum run-time required by $\mathcal{A}$) over all $D\in\mathcal{C}$, and over all internal randomness of the learner. Both an  $(\epsilon,\delta,O,d)$-PAC $\gen$-learner or an $(\epsilon,\delta,O, d)$-PAC $\qgen$-learner for a concept class $\mathcal{C} \subseteq \mathcal{F}_n$ is called efficient if its time complexity is $\mathcal{O}(\mathrm{poly}(n,1/\delta,1/\epsilon))$.
\end{definition}
\noindent Additionally, we can define both the efficient \textit{classical} PAC generator learnability of a distribution class (Definition \ref{d:pac_gen_concept_classical}) as well as the efficient \emph{quantum} PAC generator-learnability of a distribution class (Definition~\ref{d:pac_gen_concept_quantum}):
\begin{definition}[Efficient Classical PAC Generator-Learnability of a Distribution Class]\label{d:pac_gen_concept_classical}
We say that a distribution class $\mathcal{C}$ is efficiently  classically PAC $\gen$-learnable ($\qgen$-learnable) with respect to oracle $O$ and distance measure $d$ if for all $\epsilon > 0$, $ 0< \delta < 1$ there exists an efficient classical $(\epsilon,\delta,O,d)$-PAC $\gen$-learner ($\qgen$-learner) for $\mathcal{C}$.
\end{definition}
\begin{definition}[Efficient Quantum PAC Generator-Learnability of a Distribution Class]\label{d:pac_gen_concept_quantum}
We say that a distribution class $\mathcal{C}$ is efficiently quantum PAC $\gen$-learnable ($\qgen$-learnable) with respect to oracle $O$ and distance measure $d$ if for all $\epsilon > 0$, $ 0< \delta < 1$ there exists an efficient quantum $(\epsilon,\delta,O,d)$-PAC $\gen$-learner ($\qgen$-learner) for $\mathcal{C}$.
\end{definition}
\noindent Given these definitions, we are finally in a position to state precisely the primary question that we explore in this work:
\begin{question}\label{q:sample_vs_sample}
Does there exist a distribution class $\mathcal{C}$, which can be efficiently classically generated, and which 
\begin{enumerate}[label=(\alph*)]
    \item is not efficiently classically PAC $\gen$-learnable with respect to the $\sample$ oracle and the TV-distance,
    \item is efficiently quantum PAC $\gen$-learnable with respect to the $\sample$ oracle and the TV-distance. 
\end{enumerate}
\end{question}

\begin{figure}
		\centering
		\includegraphics[width=\linewidth]{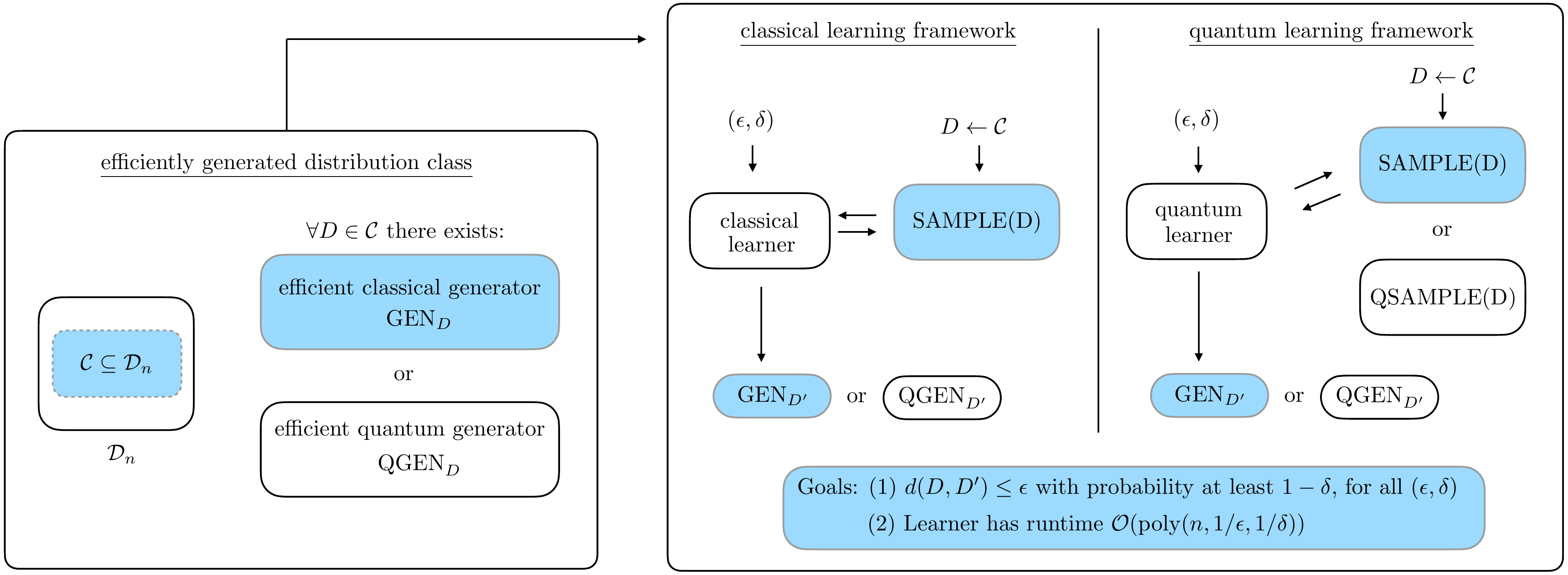}
		\caption{An overview of the PAC framework for generative modelling. As discussed in the main text, one could consider either classically generated or quantum generated distribution classes, quantum learners with access to either classical or quantum sample oracles, and learners which output either classical or quantum hypothesis generators. In this work, we focus on the setting indicated by the shaded boxes -- i.e. distribution classes which are efficiently classically generated, quantum learners with access only to classical samples, and both classical and quantum learners which output classical hypothesis generators.}
		\label{f:gen_pac_fig}
\end{figure}
	
\noindent In the following section we will show, via an explicit construction of such a distribution class,  that up to a standard cryptographic assumption, the answer to this question is ``Yes". Before continuing though, it is again worth clarifying a few potentially subtle aspects of the above question (also illustrated in Fig. \ref{f:gen_pac_fig}):\newline

\noindent\noindent\textbf{(1) Classically generated distribution class:} We have restricted our attention in Question \ref{q:sample_vs_sample} to distribution classes which can be \textit{efficiently classically generated} -- i.e. distribution  classes $\mathcal{C}$ with the property that for all $D \in \mathcal{C}$ there exists an efficient \textit{classical} generator $\gen_D$. We note that in principle this restriction is not necessary and if at a high level one's goal is simply to construct a generative modelling problem (i.e. distribution class) which is solvable using quantum resources (i.e. either efficiently quantum PAC $\gen$-learnable or efficiently quantum PAC $\qgen$-learnable), but not efficiently solvable with purely classical resources (i.e. not efficiently classically PAC $\gen$-learnable), then in principle this restriction is not necessary, and one could indeed also consider concept classes which are efficiently quantum generated. However, we have chosen to consider this additional constraint here both in order to make clear the conceptual distinction between generative modelling problems defined by underlying classical processes and generative modelling problems defined by underlying quantum processes, and to demonstrate clearly that quantum learning algorithms can obtain a clear advantage even for problems defined by underlying classical processes. However, despite the focus here on distribution classes which can be efficiently classically generated, the question of whether one can prove a quantum learning advantage for distribution classes which are specified by efficient quantum generators (such as those used for the demonstration of ``quantum computational supremacy" \cite{bremner_average-case_2016, BosonSampling, arute2019quantum}) remains an interesting open question, which we discuss in Section \ref{s:conclusion}.\newline

\noindent\noindent\textbf{(2) Classical sample oracle:} We note that in Question \ref{q:sample_vs_sample} we have restricted both the classical and quantum learning algorithms to the classical $\sample$ oracle access to the unknown distributions. Once again, as mentioned before, this is not strictly necessary, and one could also consider quantum learners with access to the $\qsample$ oracle. However, we have chosen to restrict ourselves here to classical $\sample$ oracle access both because this is the most natural abstraction of a typical applied generative modelling problem, and because this provides the ``fairest" playing field on which to compare the power of classical and quantum learners. That said, understanding the additional power that quantum samples might offer a quantum generator-learner is indeed also an interesting open question, which we formulate and  discuss in Section \ref{s:conclusion}.\newline

\noindent\textbf{(3) Classical output generators:} Given that we have restricted ourselves to concept classes which can be efficiently classically generated, it is natural as a first step to consider both classical and quantum $\gen$-learners - i.e. learning algorithms which are restricted to outputting classical generators. It is however also of great interest to determine whether the answer to Question \ref{q:sample_vs_sample} is still ``yes" if one considers quantum $\qgen$-learners, as most recent proposals for quantum generative modelling algorithms are of this type. Once again we discuss this possibility further in Section \ref{s:conclusion}.\newline

\noindent\noindent\textbf{(4) Total variation distance:} In order to motivate our choice of the TV-distance in Question \ref{q:sample_vs_sample} we note that Pinkser's inequality (Eq. \eqref{e:pinsker}) implies that if a distribution class is efficiently PAC $\gen$-learnable with respect to the KL-divergence, then it is efficiently PAC $\gen$-learnable with respect to the TV distance. As a result, if a concept class is \textit{not} efficiently classically PAC $\gen$-learnable with respect to the TV-distance, then it is \textit{not} efficiently classically PAC $\gen$-learnable with respect to KL-divergence. Given this, providing a positive answer to Question \ref{q:sample_vs_sample} with the TV-distance provides a stronger result then if one were to use the KL-divergence. In particular, if one were to only prove that a concept class was not PAC $\gen$-learnable with respect to the KL-divergence, this would not rule out its efficient PAC $\gen$-learnability with respect to the TV-distance.\newline

\noindent Finally, in addition to the above points, we note that Ref.\ \cite{coyle2019born} has defined ``quantum learning supremacy" as the existence of a distribution class for which, for some distance measure $d$ (not necessarily the total variation distance), for all $0<\delta<1$, and for some \textit{fixed} $\epsilon > 0$, there exists an efficient $(\epsilon,\delta,\sample, d)$ quantum PAC generator-learner (either a $\gen$-learner or a $\qgen$-learner), while for at least one value of $0<\delta<1$ there does not exist an efficient classical $(\epsilon,\delta,\sample, d)$ PAC $\gen$-learner. As such we see that a positive answer to Question \ref{q:sample_vs_sample} provides not only a clear example of what Ref.~\cite{coyle2019born} has called ``quantum learning supremacy", but in fact something slightly stronger, as a result of the fact that in order for a distribution class to be efficiently quantum PAC $\gen$-learnable (as per Definition \ref{d:pac_gen_concept_quantum}) there should exist an efficient $(\epsilon,\delta,\sample, d)$ quantum PAC $\gen$-learner not just for some fixed $\epsilon > 0$, but for \textit{all} $\epsilon > 0$. 

\section{A Quantum/Classical Distribution Learning Separation}\label{s:Seperation}

In this section we answer Question \ref{q:sample_vs_sample} in the affirmative, providing the main result of this work. To this end, we use the result of KMRRSS \cite{Kearns:1994:LDD:195058.195155} as a starting point, who have shown that any pseudorandom function (PRF) can be used to construct a distribution class which is not efficiently classically PAC $\gen$-learnable, with respect to $\sample$ and the KL-divergence. In addition, each distribution in the concept class defined by KMRRSS admits an efficient classical generator, which is fully specified by a key of the underlying PRF.
In light of this, we begin by strengthening the above result to hold also for the TV-distance. We then design a keyed function which, under the decisional Diffie-Hellman (DDH) assumption for the group family of quadratic residues \cite{boneh1998decision}, is pseudorandom from the perspective of classical adversaries, but not pseudorandom from the perspective of quantum adversaries, who in addition to distinguishing keyed instances of the function from random with membership queries, can also learn the secret key using only random examples. 
Instantiating a slight modification of the KMRRSS construction with this DDH based PRF yields a distribution class which answers Question \ref{q:sample_vs_sample} in the affirmative, under the DDH assumption for quadratic residues.

We proceed by introducing all the necessary cryptographic primitives in Section \ref{ss:crypto}. Equipped with these preliminaries, we then present in Section \ref{ss:kearns} the classical hardness result of KMRRSS \cite{Kearns:1994:LDD:195058.195155}, along with some important corollaries and modifications. Finally, given this result, in Section \ref{ss:DDH_result} we use the DDH assumption to explicitly construct a distribution class which, due to the results in Section \ref{ss:kearns}, is provably not efficiently classically learnable, but for which we are able to construct explicitly an efficient quantum learner.

\subsection{Cryptographic Primitives}\label{ss:crypto}

We begin here with a brief overview of the cryptographic notions which are necessary to understand the constructions in the following sections. For a more detailed introduction to these concepts and constructions, we refer to Ref.\ \cite{goldreich2007foundations}. The first notion that we need is that of a parameterization set.

\begin{definition}[Parameterization Set]\label{d:param_set} We call some infinite set $\mathcal{P}$ a parameterization set if: 
\begin{enumerate}
    \item $\mathcal{P}$ is the union of countably many pairwise disjoint sets, i.e. $\mathcal{P} = \bigcup_{n \in \mathbb{N}}\mathcal{P}_n$ with $\mathcal{P}_i\cap \mathcal{P}_j = \emptyset$ for all $i\neq j$.
    \item There exists an efficient (possibly probabilistic) instance generation algorithm $\mathcal{IG}$ which, for all $n\in\mathbb{N}$, on input $1^n$ outputs some $P\in\mathcal{P}_n$.
    \item There exists an efficient algorithm which for all $P\in\mathcal{P}$, on input $P$ outputs the unique $n\in\mathbb{N}$ such that $P\in\mathcal{P}_n$.
 \end{enumerate}
\end{definition}
\noindent We note that two particularly simple ``textbook" examples of parameterization sets, useful for gaining intuition, are $\mathcal{P} = \mathbb{N}$ and $\mathcal{P} = \{p\in\mathbb{N}\,|\, p \text{ is prime}\}$. In the former case one has $\mathcal{P}_n = \{n\}$, along with the deterministic instance generation algorithm $\mathcal{IG}(1^n) = n$, and in the latter case one can define $\mathcal{P}_n$ as the set of $n$-bit primes, and use as an instance generation algorithm existing efficient algorithms for sampling from the $n$-bit primes \cite{Katz:2007:IMC:1206501}. In particular, we note that in general, on input $1^n$, the instance generation algorithm $\mathcal{IG}$ effectively samples from some implicitly defined distribution over $\mathcal{P}_n$ -- i.e. $\mathcal{IG}(1^n)$ is a random variable taking values in $\mathcal{P}_n$. Using such parameterization sets we can then define indexed collections of efficiently computable functions:

\begin{definition}[Indexed Collection of Efficiently Computable Functions]\label{d:indexed_collection_f}
Given some parameterization set $\mathcal{P} $, we say that the set of functions $\{F_P\, | \, P \in \mathcal{P}\}$ is an indexed collection of efficiently computable functions if for all $P \in \mathcal{P}$ we have that $F_P:\mathcal{X}_P \rightarrow \mathcal{Y}_P$, and there exists:
\begin{enumerate}
    \item An efficient instance description algorithm which, for all $P\in\mathcal{P}$, on input $P$ outputs a description of the domain $\mathcal{X}_P$ and the codomain $\mathcal{Y}_P$.
    \item An efficient evaluation algorithm which, for all $P\in \mathcal{P}$ and all $x \in \mathcal{X}_P$, on input $P$ and $x$, outputs $F_P(x)$.
\end{enumerate}
\end{definition}

\noindent Using this, we are then able to define the notion of a collection of pseudorandom generators as follows:

\begin{definition}[Collection of Pseudorandom Generators]\label{d:PRG}
An indexed collection of efficiently computable functions $\{G_P\}$ is called a collection of pseudorandom generators if for all classical probabilistic polynomial time algorithms $\mathcal{A}$, all polynomials $p$ and all sufficiently large $n$ it holds that

\begin{equation}\label{eq:prg_property}
    \left|\mathrm{Pr}_{\substack{P \leftarrow \mathcal{IG}(1^n) \\
    x \leftarrow U(\mathcal{X}_P)}}[\mathcal{A}(P,G_P(x)) = 1] - \mathrm{Pr}_{\substack{P \leftarrow \mathcal{IG}(1^n) \\
    y \leftarrow U(\mathcal{Y}_P)}}[\mathcal{A}(P,y) = 1]\right| < \frac{1}{p(n)}
\end{equation}
where $U(\mathcal{X})$ denotes the uniform distribution over the set $\mathcal{X}$, and $\mathcal{IG}$ is the instance generation algorithm for $\{G_P\}$.
\end{definition}

\noindent We would now like to define pseudorandom functions. To do this we will need the slightly modified notion of an efficiently computable indexed collection of \textit{keyed} functions:
% \begin{definition}[Indexed Collection of Efficiently Computable Keyed Functions]\label{d:indexed_collection_kf}
% Given some infinite set of parameterizations $\mathcal{P}$, we call a collection of functions $\{F_P\,|\,P\in \mathcal{P}\}$ an indexed collection of efficiently computable keyed functions if for all $P \in \mathcal{P}$ we have that $F_P: \mathcal{K}_P\times\tilde{\mathcal{D}}_P \rightarrow \mathcal{D}'_P$, and there exists:
% \begin{enumerate}
%     \item An efficient (possibly probabilistic) instance generation algorithm $\mathcal{IG}$, which on input $1^n$ outputs some $P \in \mathcal{P}$, as well as a description of the the key space $\mathcal{K}_P$, domain $\tilde{\mathcal{D}}_P$, and codomain $\mathcal{D}'_P$.
%     \item An efficient probabilistic key selection algorithm which on input $P$ can sample efficiently from the uniform distribution over $\mathcal{K}_P$.
%     \item An efficient evaluation algorithm which for all $P\in \mathcal{P}$, all $k \in \mathcal{K}_P$ and all $x \in \tilde{\mathcal{D}}_P$, on input $P,k,x$ outputs $F_P(k,x)$.
% \end{enumerate}
% \end{definition}
\begin{definition}[Indexed Collection of Efficiently Computable Keyed Functions]\label{d:indexed_collection_kf}
Given some parameterization set $\mathcal{P}$, we call a collection of functions $\{F_P\,|\,P\in \mathcal{P}\}$ an indexed collection of efficiently computable keyed functions if for all $P \in \mathcal{P}$ we have that $F_P: \mathcal{K}_P\times\mathcal{X}_P \rightarrow \mathcal{Y}_P$, and there exists:
\begin{enumerate}
    \item An efficient instance description algorithm which, for all $P\in\mathcal{P}$, on input $P$ outputs a description of the key space $\mathcal{K}_P$, effective domain $\mathcal{X}_P$, and codomain $\mathcal{Y}_P$.
    \item An efficient probabilistic key selection algorithm which, for all $P\in\mathcal{P}$, on input $P$ can sample efficiently from the uniform distribution over $\mathcal{K}_P$.
    \item An efficient evaluation algorithm which, for all $P\in \mathcal{P}$, all $k \in \mathcal{K}_P$ and all $x \in \mathcal{X}_P$, on input $P,k,x$ outputs $F_P(k,x)$.
\end{enumerate}
\end{definition}

\noindent Given this, following Refs. \cite{Zhandry:2012:CQR:2417500.2417838,bogdanov2017pseudorandom}, we can define various types of pseudorandom function collections via the following:
\begin{definition}[Classical-Secure, Weak-Secure, Standard-Secure and Quantum-Secure Pseudorandom Function Collection]\label{d:s_ss_qs_PRF}
An indexed collection of efficiently computable keyed functions $\{F_P\}$ is called a (a) classical-secure (b) weak-secure (c) standard-secure or (d) quantum-secure pseudorandom function collection if for all (a,b) classical probabilistic (c,d) quantum polynomial time adversaries $\mathcal{A}$, all polynomials $p$, and all sufficiently large $n$, it holds that
\begin{equation}\label{eq:prf_property}
    \left|\mathrm{Pr}_{\substack{P \leftarrow\mathcal{IG}(1^n)\\
    k \leftarrow U(\mathcal{K}_P)}}[\mathcal{A}^{O(F_P(k,\cdot))}(P) = 1] - \mathrm{Pr}_{\substack{P \leftarrow\mathcal{IG}(1^n)\\
    R \leftarrow U(F:\mathcal{X}_P \rightarrow \mathcal{Y}_P)}}[\mathcal{A}^{O(R)}(P) = 1]\right| < \frac{1}{p(n)}
\end{equation}
where $U(F:\mathcal{X}_P \rightarrow \mathcal{Y}_P)$ denotes the uniform distribution over all functions from $\mathcal{X}_P$ to $\mathcal{Y}_P$ and $\mathcal{A}$ is given oracle access to (a,c) $O(f) = \mq(f)$, (b) $O(f) = \pex(f,U)$ or (d) $O(f) = \qmq(f)$.
\end{definition}
\noindent In order to clarify the above definition, we summarize informally below, using the abbreviation ``AECA" for ``all efficient classical algorithms" and the abbreviation ``AEQA" for ``all efficient quantum algorithms":
\begin{align}
    &\text{AECA with classical random example oracle access satisfy Eq. \eqref{eq:prf_property}} \implies \text{ weak-secure}. \nonumber \\
    &\text{AECA with classical membership query oracle access satisfy Eq.\ \eqref{eq:prf_property}} \implies \text{ classical-secure}. \nonumber \\
    &\text{AEQA with classical membership query oracle access satisfy Eq.  \eqref{eq:prf_property}} \implies \text{ standard-secure}. \nonumber \\
    &\text{AEQA with quantum membership query oracle access satisfy Eq. \eqref{eq:prf_property}} \implies \text{ quantum-secure}. \nonumber 
\end{align}
While at first glance the above naming conventions may seem extremely confusing, we note that if one assumes the existence of quantum computers, then the ``standard" setting in which one would like to prove pseudorandomness of a function collection -- i.e. the setting which corresponds to most realistic physical scenarios -- is the setting in which any possible adversary (including quantum adversaries) has classical membership query access to the unknown functions \cite{Zhandry:2012:CQR:2417500.2417838}.

\subsection{Classical Hardness from Classical-Secure Pseudorandom Functions}\label{ss:kearns}

Given the preliminaries from the previous section, we present below -- in Theorem \ref{t:kearns} -- a construction due to KMRRSS\ \cite{Kearns:1994:LDD:195058.195155}, which allows one to use (almost) any classical-secure pseudorandom function collection to construct a distribution class -- in fact, infinitely many such classes -- which were proven in Ref.\ \cite{Kearns:1994:LDD:195058.195155} to be not efficiently classically $\gen$-learnable, with respect to the $\sample$ oracle and the KL-divergence. Before presenting this construction, however, a few remarks are necessary. Firstly, we note that Theorem \ref{t:kearns} as presented below, is in fact both a slight generalization and strengthening of the original result from Ref.~\cite{Kearns:1994:LDD:195058.195155}. More specifically, Theorem \ref{t:kearns} below makes it explicit that (a) one can use classical-secure pseudorandom function collections parameterized by arbitrary parameterization sets (as opposed to simply $\mathcal{P} = \mathbb{N}$), provided the domain and co-domain satisfy mild requirements, and (b) this construction actually results in distribution classes which are not efficiently classically $\gen$-learnable with respect to the $\sample$ oracle and the \textit{TV-distance}. While the motivation for strengthening the result is clear, the generalization will be necessary for us, as in the following section we wish to instantiate this distribution class construction using a concrete pseudorandom function candidate, based on the DDH assumption. Additionally, in this work we wish to construct a distribution class which is not only provably hard to learn classically, but which is also provably efficiently quantum learnable. To do this we will require another modification of the construction from Ref.~\cite{Kearns:1994:LDD:195058.195155}, which is presented as Corollary~\ref{c:kearns}, and whose significance will be discussed at length in the following section.  Finally, we note that KMRRSS\ have provided in Ref.\ \cite{Kearns:1994:LDD:195058.195155} only a sketch of a proof that their construction yields distribution classes which are classically hard to learn. As we both generalize and strengthen this result, as well as ultimately require a modification (Corollary \ref{c:kearns}) of this construction, we provide here a full proof for Theorem \ref{t:kearns}, based on the original sketch from Ref.~\cite{Kearns:1994:LDD:195058.195155}.

At this stage we are almost ready to present the construction, in a language sufficiently general for our requirements. As a final preliminary consideration, we note that for all non-negative integers $x \in \mathbb{N}^0$ we will denote by $\bin(x)$ the shortest possible binary representation of $x$, and by $\bin_n(x)$ the $n$-bit binary representation obtained by padding $\bin(x)$ with zeros. We also denote by $x||y$ the concatenation of bit strings $x$ and $y$. Additionally, for any set $X \subset \mathbb{N}^0$ we write $X\subseteq \{0,1\}^n$ when $\bin_n(x)$ exists for all $x \in X$. Given these definitions we state the following theorem, which is a reformulation, generalization and strengthening of the original result from KMRRSS~\cite{Kearns:1994:LDD:195058.195155}:

\begin{theorem}[Classical Hardness from Classical-Secure Pseudorandom Functions]\label{t:kearns}
Let $\{F_P\}$ be a classical-secure pseudorandom function collection with the property that for all $n$, for all $P \in \mathcal{P}_n$, it is the case that $F_P:\mathcal{K}_P\times\{0,1\}^n \rightarrow \mathcal{Y}_P$, with $\mathcal{Y}_P \subseteq \{0,1\}^n$. For all $P$, and all $k \in \mathcal{K}_P$, we then define 
\begin{equation}
    \kgen_{(P,k)}:\{0,1\}^n \rightarrow \{0,1\}^{2n}
\end{equation}
via
\begin{equation}
    \kgen_{(P,k)}(x) = x||\bin_n(F_P(k,x)).
\end{equation}
Additionally, we denote by $\tilde{D}_{(P,k)}$ the discrete distribution over $\{0,1\}^{2n}$ for which $\kgen_{(P,k)}$ is a classical generator. For all sufficiently large $n$ the distribution class $\tilde{\mathcal{C}}_n := \{\tilde{D}_{(P,k)}|P \in \mathcal{P}_n,k\in\mathcal{K}_p\}$ is not efficiently classically PAC $\gen$-learnable with respect to the $\sample$ oracle and the TV-distance.
\end{theorem}
\noindent In order to simplify the presentation of the proof of Theorem \ref{t:kearns}, it will be convenient to begin with a few preliminary lemmas. The first result that we need is an alternative characterization of classical-secure pseudorandom function collections, which we develop below, and illustrate in Fig. \ref{f:PRF_equivalence}.

\begin{definition}[Polynomial Inference \cite{Goldreich:1986:CRF:6490.6503}]\label{d:poly_inference} Let $\{F_P\}$ be an indexed collection of keyed functions, and let $\mathcal{A}$ be some probabilistic polynomial time classical algorithm capable of oracle calls. On input $P\in \mathcal{P}_n$, algorithm $\mathcal{A}$ is given oracle access to $\mq(F_P(k,\cdot))$, and carries out a computation in which it queries the oracle on $x_1,\ldots,x_j \in \mathcal{X}_P$. Algorithm $\mathcal{A}$ then outputs some $x\in \mathcal{X}_P$, which must satisfy $x \notin \{x_1,\ldots,x_j$\}. We call $x$ the ``exam string". At this point, $\mathcal{A}$ is then disconnected from $\mq(F_P(k,\cdot))$ and presented the two values $F_P(k,x)$ and $y \leftarrow U(\mathcal{Y}_P)$ in random order. We say that $\mathcal{A}$ ``passes the exam" if it correctly guesses which of the two values is $F_{P}(k,x)$. Let $Q$ be some polynomial. We then say that $\mathcal{A}$ $Q$-infers the collection $\{F_P\}$ if for infinitely many $n$, given input $P\in \mathcal{P}_n$, it passes the exam with probability at least $1/2 + 1/Q(n)$, where the probability is taken uniformly over all possible choices of $P\in \mathcal{P}_n$, $k\in \mathcal{K}_P$, all possible choices of $y\in \mathcal{Y}_P$ and all possible orders of the exam strings $F_P(k,x)$ and $y$. We say that an indexed collection of keyed functions $\{F_P\}$ can be polynomially inferred if there exists a polynomial $Q$ and a probabilistic polynomial time algorithm $\mathcal{A}$ which $Q$-infers $\{F_P\}$.
\end{definition}

\begin{figure}
		\centering
		\includegraphics[width=\linewidth]{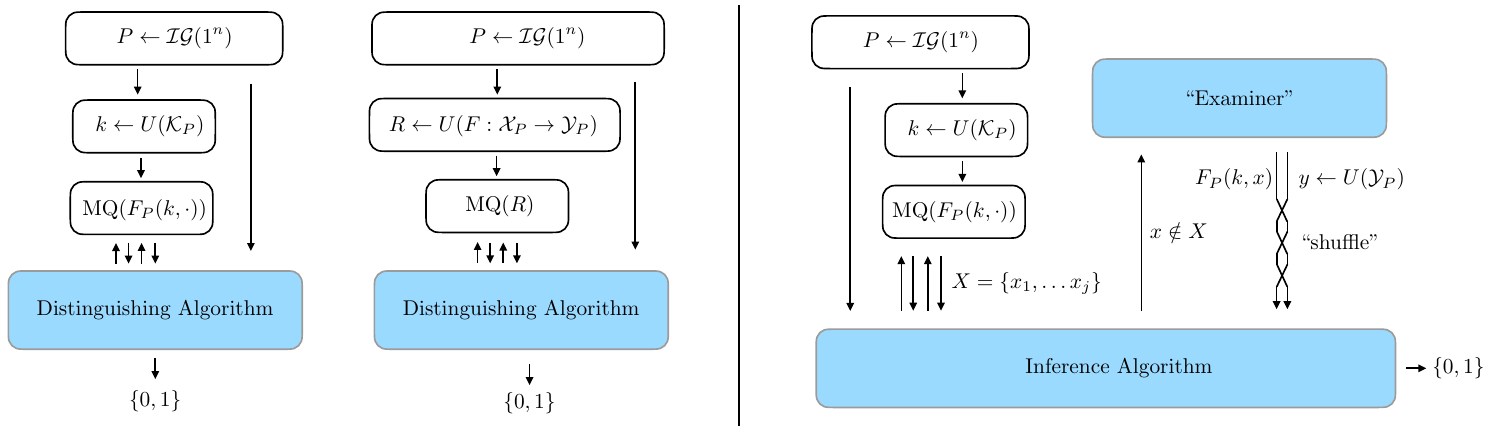}
		\caption{Illustration of equivalent notions of a classical-secure PRF collection. The left panel illustrates the setting as per Definition \ref{d:s_ss_qs_PRF}, in which we consider distinguishing algorithms which, with membership query access, try to distinguish between a randomly drawn instance of the keyed function collection $F_{P}(k,\cdot)$ and a function $R$ drawn uniformly at random. The right panel illustrates Definition \ref{d:poly_inference}, in which we consider an inference algorithm, which after a learning phase, should try to pass an ``exam" of its own choosing. As per Lemma \ref{l:gold_predictability}, for a given indexed collection of keyed functions, there exists a suitable distinguishing algorithm, if and only if there exists a suitable inference algorithm.  } 
		\label{f:PRF_equivalence}
\end{figure}

\begin{lemma}[\cite{Goldreich:1986:CRF:6490.6503}]\label{l:gold_predictability}
Let $\{F_P\}$ be an indexed collection of efficiently computable keyed functions. Then, $\{F_p\}$ cannot be polynomially inferred if and only if it is a classical-secure pseudorandom function collection.
\end{lemma}
\noindent Additionally, we will need the following observation.

\begin{lemma}\label{l:proof_inter_lemma}
Let $\gen_D$ be a $(d_{\mathrm{TV}},\epsilon)$-generator for $\tilde{D}_{(P,k)}$, for some $\epsilon < 1/5$. 
Then, for at least $2^n/2$ of the $2^n$ possible strings of the form $y = x||\bin_n(F_p(k,x)) \in \{0,1\}^{2n}$ with $x\in\{0,1\}^n$ it is the case that
\begin{equation}
    D(x||\bin_n(F_{P}(k,x))) 
    \geq \frac{\epsilon}{2^n}.
\end{equation}
\end{lemma}
\begin{proof}
Assume that the claim is false, i.e. that for at least $2^n/2$ of the strings $y = x||\bin_n(F_p(k,x)) \in \{0,1\}^{2n}$  one has that $D(x||\bin_n(F_{P}(k,x))) < \epsilon/2^n$. Then
\begin{align}
    d_{\mathrm{TV}}(D,\tilde{D}_{(P,k)}) &= \frac{1}{2}\sum_{y\in\{0,1\}^{2n}}|\tilde{D}_{(P,k)}(y)-D(y)| \\
    &= \frac{1}{2}\sum_{x\in\{0,1\}^{n}}|\frac{1}{2^n}-D(x||\bin_n(F_p(k,x)))| + \frac{1}{2}\sum_{\text{other } y}|D(y)| \\
    &\geq  \frac{1}{2}\sum_{x\in\{0,1\}^{n}}|\frac{1}{2^n}-D(x||\bin_n(F_p(k,x)))|  \\
    & > \frac{1}{2}\left(\frac{2^n}{2}\left[\frac{1}{2^n}-\frac{\epsilon}{2^n}\right]\right)\\
    &= \frac{1}{4}(1-\epsilon)\\
    & > \epsilon.        \qquad \qquad\qquad\qquad \qquad\qquad\qquad \qquad\qquad\qquad \qquad(\text{via } \epsilon < 1/5)
\end{align}
This contradicts the assumption that $\gen_D$ is a $(d_\mathrm{TV},\epsilon)$-generator for $\tilde{D}_{(P,k)}$.
\end{proof}
\noindent Finally, given these preliminary results and observations, we can present a full proof of Theorem \ref{t:kearns}.

\begin{proof}[Proof (Theorem \ref{t:kearns})] At a high level, the idea of the proof is to assume that $\tilde{C}_n$ is efficiently classically PAC $\gen$-learnable for infinitely many $n$, and use the associated learning algorithms to construct a poly-time algorithm which $Q$-infers $\{F_P\}$, for some polynomial $Q$. By Lemma \ref{l:gold_predictability} this implies that $\{F_P\}$ is not classical-secure pseudorandom, which then gives a proof by contradiction. To do this, let us denote the assumed learning algorithm for $\tilde{C}_n$ as $\mathcal{\tilde{A}}_n$. Our goal is now to construct an inference algorithm for $\{F_P\}$, which we denote as $\mathcal{A}$. Now, as per Definition \ref{d:poly_inference}, on input $P\in \mathcal{P}_n$, when given access to $\mq(F_{P}(k,\cdot))$, algorithm $\mathcal{A}$ proceeds in two steps as follows:

\begin{enumerate}
    \item Obtain an approximate generator for $\tilde{D}_{(P,k)}$ by simulating the learning algorithm $\tilde{\mathcal{A}}_n$. Specifically, run the learning algorithm $\tilde{\mathcal{A}}_n$, with $\epsilon = 1/n$ and $\delta = 1/2$, by using access to $\mq(F_{P}(k,\cdot))$ to simulate $\sample(\tilde{D}_{(P,k)})$.  Each time $\tilde{\mathcal{A}}_n$ queries $\sample(\tilde{D}_{(P,k)})$, algorithm $\mathcal{A}$ simply draws some $x\in \{0,1\}^n$ uniformly at random, queries  $\mq(F_{P}(k,\cdot))$ on input $x$, and then provides $\tilde{\mathcal{A}}_n$ with the sample
    \begin{equation}
       x||\bin_n[\query[\mq(F_P(k,\cdot)](x)] =  x||\bin_n(F_{P}(k,x)) = \kgen_{(P,k)}(x).
    \end{equation}
    Let us denote by $\gen_D$ the generator output by $\tilde{\mathcal{A}}_n$, and by $X = \{x_1,\ldots,x_j\}$ the set of strings used by $\mathcal{A}$ to simulate $\tilde{\mathcal{A}}_n$. We know that with probability $1-\delta = 1/2$, the output generator $\gen_D$ is a $(d_{\mathrm{TV}},\epsilon)$-generator for $\tilde{D}_{(P,k)}$. Additionally, it follows from the efficiency of $\mathcal{A}_n$ that $|X| = \mathrm{poly}(n,1/\delta,1/\epsilon) = \mathrm{poly}(n)$.
    \item Use $\gen_D$ to generate a sample $x||y \in \{0,1\}^{2n}$ from $D$.
    \begin{itemize}
    \item If $x\notin X$, then submit $x$ as the exam string, and receive the strings $y_1,y_2$. If $y\in \{y_1,y_2\}$, then output $y$. Let us call this case-$\encircle{$a$}$. Else, if $y\notin \{y_1,y_2\}$ then output either $y_1$ or $y_2$ 
    uniformly at random. We call this case-$\encircle{b}$.
    \item Else, if $x\in X$, then select any $\tilde{x} \notin X$ as the exam string, and after receiving $y_1,y_2$ simply output either $y_1$ or $y_2$ at random. Call this case-$\encircle{c}$.  
    \end{itemize}
\end{enumerate}
We now want to determine a lower bound on the probability that $\mathcal{A}$ passes the exam. To do this, let us denote by $\mathrm{Pr}[\encircle{z}]$ the probability that case-$\encircle{z}$ occurs, and by $\mathrm{Pr}_{\mathrm{pass}}[\encircle{z}]$ the conditional probability that $\mathcal{A}$ passes the exam, given that case-$\encircle{z}$ has
occurred. Clearly, $\mathrm{Pr}_{\mathrm{pass}}[\encircle{b}]=\mathrm{Pr}_{\mathrm{pass}}[\encircle{c}]=1/2$, so let us look at case-$\encircle{a}$ more carefully. In particular, there are two possibilities:
\begin{enumerate}
    \item The first possibility is that $y = F_{P}(k,x)$. Let's call this case-$\encircle{a1}$. In this case $\mathcal{A}$ definitely passes the exam -- i.e. we have that $\mathrm{Pr}_{\mathrm{pass}}[\encircle{a1}] = 1$.
    \item The second possibility is that $y\neq F_P(k,x)$, and that whichever string from $\{y_1,y_2\}$ was randomly drawn, just happens to equal $y$. Lets call this case-$\encircle{a2}$. In this case $\mathcal{A}$ definitely fails the exam -- i.e. $\mathrm{Pr}_{\mathrm{pass}}[\encircle{a2}] = 0$.
\end{enumerate}
In light of this, the probability that $\mathcal{A}$ passes the exam is then given by 
\begin{align}
    \mathrm{Pr}_{\mathrm{pass}} &= \mathrm{Pr}_{\mathrm{pass}}[\encircle{a1}]\mathrm{Pr}[\encircle{a1}] + \mathrm{Pr}_{\mathrm{pass}}[\encircle{a2}]\mathrm{Pr}[\encircle{a2}] +  \mathrm{Pr}_{\mathrm{pass}}[\encircle{b}]\mathrm{Pr}[\encircle{b}] + \mathrm{Pr}_{\mathrm{pass}}[\encircle{c}]\mathrm{Pr}[\encircle{c}] \nonumber \nonumber\\
    &= \mathrm{Pr}[\encircle{a1}] + \frac{1}{2}\left(\mathrm{Pr}[\encircle{b}] + \mathrm{Pr}[\encircle{c}]\right) \nonumber \\
    &= \mathrm{Pr}[\encircle{a1}] + \frac{1}{2}\left(1-  \mathrm{Pr}[\encircle{a1}] - \mathrm{Pr}[\encircle{a2}]\right) \nonumber \\
    &= \frac{1}{2} + \frac{1}{2}\mathrm{Pr}[\encircle{a1}] - \frac{1}{2}\mathrm{Pr}[\encircle{a2}].
\end{align}
So, to proceed we now analyze $\mathrm{Pr}[\encircle{a1}]$ and $\mathrm{Pr}[\encircle{a2}]$. Notice that case-$\encircle{a1}$ occurs when $x||y = x||F_P(k,x)$ for some $x\notin X$. Using Lemma \ref{l:proof_inter_lemma}, with $\epsilon = 1/n$, and $n> 5$ (so that $\epsilon < 1/5$) we know that, when $\gen_D$ is a 
$(d_{\mathrm{TV}},\epsilon)$-generator for $\tilde{D}_{(P,k)}$, there exist at least $2^n/2$ strings of the form $x||F_P(k,x)$ for which 
\begin{equation}
    D(x||F_P(k,x)) \geq \frac{1}{n2^n}.
\end{equation}
Using the above, along with the fact that $|X| = p(n)$ for some polynomial $p$, we then have that 
\begin{align}
    \mathrm{Pr}[\encircle{a1}] &\geq \mathrm{Pr}[\encircle{a1}\,|\,d_{\mathrm{TV}}(D,\tilde{D}_{(P,K)}) \leq \epsilon]\times \mathrm{Pr}[d_{\mathrm{TV}}(D,\tilde{D}_{(P,K)}) \leq \epsilon] \nonumber \\
    &\geq\left[ \left(\frac{2^n}{2} - p(n)\right) \frac{1}{n2^n}\right]\times\frac{1}{2}\nonumber\\
    &\geq \frac{1}{4n}\left(1 - \frac{2p(n)}{2^n}\right).
\end{align}
As a result, there exists some $n_1$ such that for all $n \geq n_1$ we have that $\mathrm{Pr}[\encircle{a1}] \geq 1/(6n)$. So, at this point we know that for all $n$ large enough
\begin{align}
    \mathrm{Pr}_{\mathrm{pass}} \geq \frac{1}{2} + \frac{1}{6n} - \frac{1}{2}\mathrm{Pr}[\encircle{a2}].
\end{align}
Now, note that case-$\encircle{a2}$ occurs when $x\notin X$ and when whichever of $y_1$ or $y_2$ is randomly drawn is equal to $y$. As a result, we have that $\mathrm{Pr}[\encircle{a2}] \leq 1/2^n$. Using this, we see that for all $n\geq n_1$, 
\begin{align}
    \mathrm{Pr}_{\mathrm{pass}} &\geq \frac{1}{2} + \frac{1}{6n} - \frac{1}{2\times2^n} \nonumber \\
    &\geq \frac{1}{2} + \frac{1}{6n}\left(1 - \frac{3n}{2^n}\right).
\end{align}
Similarly, to the previous case, we now know that there exists some $n_2$, such that for all $n \geq \max\{n_1,n_2\}$,
\begin{align}
    \mathrm{Pr}_{\mathrm{pass}}
    &\geq \frac{1}{2} + \frac{1}{7n} \nonumber \\
    &:= \frac{1}{2} + \frac{1}{Q(n)}.
\end{align}
In light of the above, we therefore see that for all sufficiently large $n$, $\mathcal{A}$ $Q$-infers $\{F_p\}$, and therefore, via Lemma \ref{l:gold_predictability}, $\{F_P\}$ cannot be classical-secure pseudorandom, which contradicts the assumptions of the theorem.
\end{proof}
As mentioned earlier, while Theorem \ref{t:kearns} provides a method for the construction of distribution classes which are not efficiently classically learnable, in order to construct such a distribution class which is also efficiently quantum learnable, it will be helpful to formulate the following modified construction:

\begin{corollary}\label{c:kearns}
Let $\{F_P\}$ be a classical-secure pseudorandom function satisfying all the properties required for Theorem \ref{t:kearns}. In addition, for all $n$, we assume that for all $P \in \mathcal{P}_n$ there exists an efficient $m = \mathrm{poly}(n)$ bit encoding of $P$, which we denote as $\bin_m(P)$. For all $P$, and all $K \in \mathcal{K}_P$ we then define
\begin{equation}
    \gen_{(P,k)}:\{0,1\}^n \rightarrow \{0,1\}^{2n + m}
\end{equation}
via
\begin{equation}
    \gen_{(P,k)}(x) = x||\bin_n(F_P(k,x))||\bin_m(P).
\end{equation}
Additionally, we define $D_{(P,k)}$ as the discrete distribution over $\{0,1\}^{2n+m}$ for which $\gen_{(P,k)}$ is a classical generator. For all sufficiently large $n$ the distribution class $\mathcal{C}_n := \{D_{(P,k)}|P \in \mathcal{P}_n,k\in\mathcal{K}_p\}$ is not efficiently classically PAC $\gen$-learnable with respect to the $\sample$ oracle and the TV-distance.
\end{corollary}
\noindent To make clear the difference between the constructions of Theorem \ref{t:kearns} and Corollary \ref{c:kearns} we summarize informally as follows:
\begin{align}
\text{Theorem \ref{t:kearns}}&\rightarrow \kgen_{(P,k)}(x) = x||\bin_n(F_P(k,x)) , \nonumber \\
    \text{Corollary \ref{c:kearns}}&\rightarrow \gen_{(P,k)}(x) = x||\bin_n(F_P(k,x))||\bin_m(P) .\nonumber
\end{align}
Before describing the motivation behind such a modification, we note that we have stated this construction as a corollary due to the fact that the proof is essentially the same as the proof of Theorem \ref{t:kearns}. The only difference is that  when the polynomial inference algorithm $\mathcal{A}$ is given input $P \in \mathcal{P}_n$ and access to $\mq(F_P(k,\cdot))$, in order to simulate the learning algorithm $\tilde{\mathcal{A}}_n$, it should respond to a $\sample$ query by drawing some $x\in\{0,1\}^n$ uniformly at random, and then returning
\begin{equation}
       x||\bin[\query[\mq(F_P(k,\cdot)](x)]||\bin_m(P) =  x||\bin_n(F_{P}(k,x))||\bin_m(P) = \gen_{(P,k)}(x).
\end{equation}
To see why such a modified construction will be helpful for constructing an efficient quantum learner, note that both $\kgen_{(P,k)}$ and $\gen_{(P,k)}$ are fully specified by the parameterization $P$, and some key $k\in\mathcal{K}_P$. As such, given $\sample$ access to either generator, it would be sufficient for a generator learning algorithm to learn the tuple $(P,k)$. If one uses the distribution class $\tilde{D}_{(P,k)}$ of Theorem \ref{t:kearns}, generated by $\kgen_{(P,k)}$, then a learner really has to learn both the parameterization $P$, and the key $k$. However, if one uses the distribution class $D_{(P,k)}$ of Corollary \ref{c:kearns}, generated by $\gen_{(P,k)}$, then with every sample from the distribution the learner is \textit{given} a binary encoding of $P$, and as such only needs to learn the key $k$. 

\subsection{Quantum Learnability and Classical Hardness for a DDH Based Distribution Class}\label{ss:DDH_result}

Theorem \ref{t:kearns} and Corollary \ref{c:kearns} provide the first ingredient necessary to answer Question \ref{q:sample_vs_sample} in the affirmative; namely a technique for constructing, from almost any classical-secure pseudorandom function collection, a distribution class which can be efficiently classically generated, and which is not efficiently classically PAC $\gen$-learnable. Given this result, on first impressions one might think that all that is required for such distribution classes to be efficiently quantum PAC $\gen$-learnable is that the underlying classical-secure PRF collection is not standard-secure -- i.e. is not pseudorandom from the perspective of quantum adversaries with classical membership query access. In particular, it seems plausible that if the underlying PRF is classical-secure but not standard-secure, then one could exploit the quantum PRF adversary $\mathcal{A}'$ for the construction of a quantum generator learner $\mathcal{A}$. Unfortunately, however, it is in fact not so straightforward, for the following reasons:
Firstly, we note that 

\begin{align}
    \query[\sample(D_{(P,k)})] &= \gen_{(P,k)}(x) \text{ with } x\leftarrow U_n \nonumber \\
    &= x||\bin_n(F_P(k,x))||\bin_m(P) \text{ with } x\leftarrow U_n.
\end{align}
Additionally, we have that
\begin{equation}
    \query[\pex(F_P(k,\cdot),U_n)] = (x,F_P(k,x)) \text{ with } x\leftarrow U_n,
\end{equation}
and so by comparison we see that if a learning algorithm $\mathcal{A}$ is given access to the oracle $\sample(D_{(P,k)})$, then it can efficiently simulate oracle access to $\pex(F_P(k,\cdot),U_n)$, however, it \emph{cannot} simulate oracle access to $\mq(F_P(k,\cdot))$. 
As such even if there exists an efficient quantum adversary $\mathcal{A}'$ for the classical-secure PRF $\{F_P\}$, a learning algorithm with oracle access to $\sample(D_{(P,k)})$ could not simulate $\mathcal{A}'$, which requires access to $\mq(F_P(k,\cdot))$. Additionally, note from Eq. \eqref{eq:prf_property} that any quantum PRF adversary $\mathcal{A}'$ requires as input the corresponding parameterization $P$. As such, even if this quantum adversary could succeed with only $\pex$ access, a learning algorithm with access only to the $\sample(\tilde{D}_{(P,k)})$ oracle could not simulate $\mathcal{A}'$, as it does not have access to $P$. However, a learning algorithm with access to $\sample(D_{(P,k)})$ could simulate $\mathcal{A}'$, as an encoding of the parameterization $P$ is given in the suffix of each sample from the distribution. This in fact provides an additional important motivation for the formulation of Corollary \ref{c:kearns}.  Secondly, in order for a collection of keyed functions to not be standard-secure, all that is required is that a quantum adversary, with classical membership query access can, for all sufficiently large $n$, distinguish instances of the keyed function from randomly drawn functions with non-negligible probability, as formalized by condition~\eqref{eq:prf_property}. As such, even if the classical-secure PRF underlying the distribution class is not standard-secure, and even if it is not standard-secure with respect to $\pex$ access as opposed to $\mq$ access, this does not instantly imply that the quantum PRF adversary could be turned into a quantum distribution learner, which should for all valid $(\delta,\epsilon)$ be able to learn a $(d_{\mathrm{TV}},\epsilon)$-generator, with probability $1-\delta$. However, as per the discussion in the previous section, we know that the generators $\{\gen_{(P,k)}\}$ for the distribution class $\mathcal{C}_n = \{D_{{(P,k)}}\}$ are fully specified by the tuple $(P,k)$, and that as an encoding of $P$ is given ``for free" with each sample from the distribution, being able to learn the key $k\in \mathcal{K}_P$ from $\sample$ access to $\{D_{(P,k)}\}$ is sufficient to learn the \textit{exact} generator $\gen_{(P,k)}$. Given this, we see that if the underlying classical-secure but not standard-secure PRF is such that the quantum adversary $\mathcal{A'}$
\begin{enumerate}[label=(\alph*)]
    \item requires only $\pex$ access, as opposed to $\mq$ access,
    \item can in addition to distinguishing an instance of the keyed function from a random function, also learn the key itself with any specified probability, 
\end{enumerate} 
\begin{figure}
		\centering
		\includegraphics[width=\linewidth]{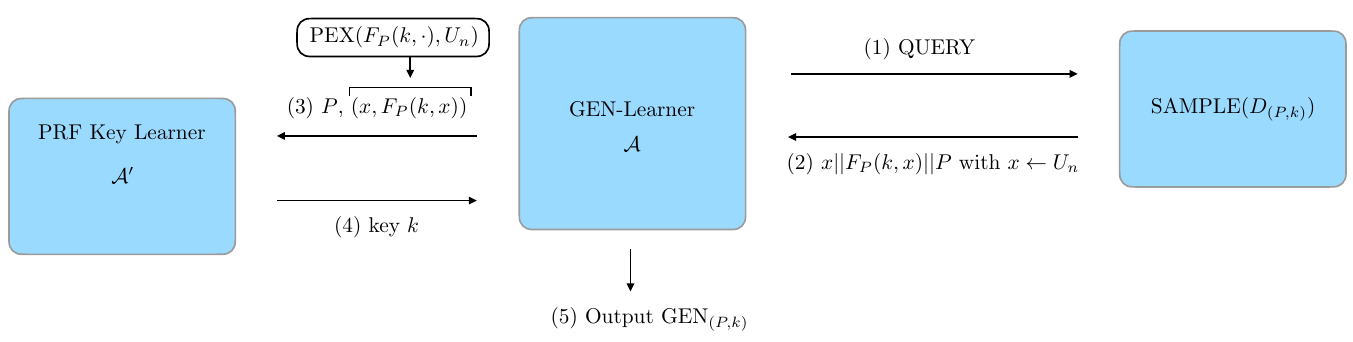}
		\caption{Overview of the strategy for proving quantum learnability. In particular, if the classical-secure PRF $\{F_P\}$ has the property that there exists an efficient quantum \textit{key learner}, an efficient quantum algorithm $\mathcal{A}'$ which on input $P$ and access to $\pex(F_P(k,\cdot),U_n)$ can learn the key $k$, then one can construct a $\gen$-learner $\mathcal{A}$ for $\{D_{(P,k)}\}$ by using oracle access to $\sample(D_{(P,K)})$ to simulate the key learner $\mathcal{A}'$.}
		\label{f:proof_sketch}
\end{figure}
then this PRF adversary could be efficiently simulated by a learner $\mathcal{A}$ with access to $\sample(D_{(P,k)})$, and used to learn $\gen_{(P,k)}$ with any desired probability (as illustrated in Fig. \ref{f:proof_sketch}).  Instantiating the construction of Corollary \ref{c:kearns} with this PRF would then guarantee that $\{D_{(P,k)}\}$ can not be efficiently classically learned, and therefore yield a distribution class which provides an affirmative answer to Question \ref{q:sample_vs_sample}. In light of this, we therefore in the remainder of this section construct a classical-secure but not standard-secure PRF with the properties listed above.

The PRF we construct will be classical-secure under the DDH assumption for the group family of quadratic residues. We use such a construction as this assumption is known not to hold for quantum adversaries \cite{boneh1998decision}. In order to present this construction, we require the following preliminary definitions, following closely the presentation of Ref.\ \cite{boneh1998decision}:

\begin{definition}[Efficient Group Family]\label{d:eff_group_family}
Given some parameterization set $\mathcal{P}$, we say that a set of groups $\mathbb{G} = \{\mathbb{G}_P \,| \, P \in \mathcal{P}\}$ is an efficient group family if for all $P\in \mathcal{P}$ the group $\mathbb{G}_P$ is a finite cyclic group, and there exists:

\begin{enumerate}
    \item An efficient instance description algorithm $\mathcal{ID}$ which for all $P\in \mathcal{P}$, on input $P$ outputs a generator $g$ for the group $\mathbb{G}_P$.
    \item An efficient group product algorithm $\mathcal{G}$ which, for all $P\in\mathcal{P}$ and for all $a,b\in\mathbb{G}_P$, on input $(P,a,b)$ outputs the group product $a*b \in \mathbb{G}_P$.
\end{enumerate}
\end{definition}

\begin{definition}[Quadratic Residues in $\mathbb{Z}^*_N$]\label{d:QR}
Let $\mathbb{Z}^*_N$ denote the multiplicative group of integers modulo $N$. We say that an element $y\in \mathbb{Z}^*_N$ is a \emph{quadratic residue modulo $N$} if there exists an $x\in \mathbb{Z}^*_N$ such that $x^2 \equiv y \mod N$. We denote the set of quadratic residues modulo $N$ by $\mathrm{QR}_N$.
\end{definition}
\noindent We note that $\mathrm{QR}_N$ forms a subgroup of $\mathbb{Z}^*_N$. Additionally, when $N$ is prime, squaring modulo $N$ is a two-to-one map, and therefore exactly half the elements of $\mathbb{Z}^*_N$ are quadratic residues. As a result, if $N$ is a \emph{safe prime} -- i.e. $N = 2q + 1$ with $q$ prime -- then $\mathrm{QR}_N$ is of prime order $q = (N-1)/2$, and therefore cyclic, with all elements (except the identity) as generators. Given this, we make the following observations:
\begin{observation}[Parameterization Set of Safe Primes]\label{o:safe_primes}
The set 
\begin{align}
    \mathcal{P_{\text{safe primes}}} &:= \{p\,|\, p \text{ is a safe prime}\} \\
    &= \bigcup_{n \in \mathbb{N}}\{p\,|\, p \text{ is an } n\text{-bit safe prime} \} \\
    &:=\bigcup_{n\in\mathbb{N}}\mathcal{P}_n
\end{align}
is a valid parameterization set. In particular, we have by definition that $\mathcal{P}_i\cap\mathcal{P}_j = \emptyset$ for all $i\neq j$, there exists an efficient probabilistic instance generation algorithm which, for all $n\in\mathbb{N}$, maps from $1^n$ to the set of $n$-bit safe primes \cite{Katz:2007:IMC:1206501}, and for all $n\in\mathbb{N}$, given some $p\in{\mathcal{P}_n}$ one can recover $n$ efficiently.
\end{observation}

\begin{observation}[Group Family of Quadratic Residues]\label{o:qr}
The set of groups $\mathbb{G} = \{\mathrm{QR}_p \,|\, p \in \mathcal{P}_{\text{safe primes}}\}$ is an efficient group family. More specifically, we have already established that $\mathcal{P}_{\text{safe primes}}$ is a valid parameterization set. Additionally, one can construct an efficient instance description algorithm from (a) the existence of an efficient algorithm for testing membership in $\mathrm{QR}_p$ \cite{blum1984generate}, and (b) the observation that all elements of $\mathrm{QR}_p$ except the identity are generators. Finally, The existence of an efficient group operation algorithm follows from the fact that modular multiplication can be performed efficiently \cite{Katz:2007:IMC:1206501}. We refer to $\{\mathrm{QR}_p\}$ as the group family of quadratic residues.
\end{observation}

\noindent With these preliminaries, we can now state the decisional Diffie-Hellman assumption:

\begin{definition}[Decisional Diffie-Hellman Assumption \cite{boneh1998decision,naor2004number}]\label{d:DDH}
We say that a group family $\mathbb{G} = \{\mathbb{G}_P\,|\,P\in\mathcal{P}\}$, with instance generator $\mathcal{IG}$ and instance description algorithm $\mathcal{ID}$, satisfies the DDH assumption, if for all probabilistic polynomial time algorithms $\mathcal{A}$, all polynomials $p(\cdot)$, and all sufficiently large $n$ it holds that
\begin{equation}
    |\mathrm{Pr}_{\substack{P \leftarrow \mathcal{IG}(1^n)\\
    g\leftarrow \mathcal{ID}(P)\\
    a,b \leftarrow \mathbb{Z}_{|\mathbb{G}_P|}}}[\mathcal{A}(P,g,g^a,g^b,g^{ab}) = 1] - \mathrm{Pr}_{\substack{P \leftarrow \mathcal{IG}(1^n)\\
    g\leftarrow \mathcal{ID}(P)\\
    a,b,c \leftarrow \mathbb{Z}_{|\mathbb{G}_P|}}}[\mathcal{A}(P,g,g^a,g^b,g^{c}) = 1]| < \frac{1}{p(n)}, \label{e:ddh_property}
\end{equation}
where the probability in the first term is taken over the random variable $\mathcal{IG}(1^n)$ (i.e. a random parameterization $P\in\mathcal{P}_n$), the random variable $\mathcal{ID}(P)$ (i.e. a random generator $g$ for the group $\mathbb{G}_P$) and $a,b \in \mathbb{Z}_{|\mathbb{G}_P|}$ selected uniformly at random, and in the second term $c \in \mathbb{Z}_{|\mathbb{G}_P|}$ is also chosen uniformly at random.
\end{definition}
\noindent From this point on we will restrict ourselves to the group family of quadratic residues parameterized by safe primes, which is believed to satisfy the DDH assumption \cite{boneh1998decision}. The first thing to note is that the DDH assumption instantly implies a method for the construction of a collection of pseudorandom generators. To show this, let us start by defining the parameterization set $\mathcal{P}_{(p,g,g^a)}$ as the infinite set of all tuples of the form $(p,g,g^a)$ where $p = 2^q+1$ is some safe prime, $g$ is a generator for $\mathrm{QR}_p$ and $a\in\mathbb{Z}_q$. We denote the subset of all such tuples in which $p$ is an $n$-bit prime as $\mathcal{P}_{n,(p,g,g^a)}$ and we note that $\mathcal{P}_{(p,g,g^a)} = \bigcup_{n\in\mathbb{N}}\mathcal{P}_{n,(p,g,g^a)}$ is a valid parameterization set. In particular, for all $n\in\mathbb{N}$, on input $1^n$ the instance generation algorithm for $\mathcal{P}_{(p,g,g^a)}$ first runs the instance generation algorithm for $\mathcal{P}_{\textit{safe primes}}$ to obtain an $n$-bit safe prime $p$, then runs the instance description algorithm for the group family $\{\mathrm{QR}_p\}$ to obtain a generator $g$ for $\mathrm{QR}_p$, before finally selecting an element $a$ of $\mathbb{Z}_q$ uniformly at random. We then define the function $\textsf{modexp}_{p,g}:\mathbb{N}\rightarrow \mathrm{QR}_p$ via
\begin{align}
    \textsf{modexp}_{p,g}(b) := g^b \md(p),
\end{align}
which allows us, for all tuples $(p,g,g^a) \in \mathcal{P}_{(p,g,g^a)}$, to define the function
\begin{equation}
    G_{(p,g,g^a)}:\mathbb{Z}_q \rightarrow \mathrm{QR}_p\times \mathrm{QR}_p
\end{equation}
via
\begin{align}
    G_{(p,g,g^a)}(b) &= \textsf{modexp}_{p,g}(b)||\textsf{modexp}_{p,g^a}(b) \nonumber \\
    &= g^a \md(p)|| g^{ab} \md(p) \nonumber\\
    & :=G^0_{(p,g,g^a)}(b)||G^1_{(p,g,g^a)}(b). \label{eq:length_doubling}
\end{align}

\noindent Given this construction, we now note, following Ref.\ \cite{boneh1998decision}, that $\{G_{(p,g,g^a)}\,|\,(p,g,g^a) \in \mathcal{P}_{(p,g,g^a)}\}$ is a collection of pseudorandom generators, under the assumption that $\{\mathrm{QR}_p\}$ satisfies the DDH assumption.

\begin{observation}[$\{G_{(p,g,g^a)}\}$ is a Collection of Pseudorandom Generators \cite{boneh1998decision}]\label{o:is_prg} Note that when using the group family of quadratic residues, we can rewrite Eq.\ \eqref{e:ddh_property} as
\begin{equation}
    |\mathrm{Pr}_{\substack{p,g,g^a \leftarrow \mathcal{IG}(1^n)\\
    b \leftarrow \mathbb{Z}_q}}[\mathcal{A}(p,g,g^a,G_{(p,g,g^a)}(b)) = 1] - \mathrm{Pr}_{\substack{p,g,g^a \leftarrow \mathcal{IG}(1^n)\\
    b,c \leftarrow \mathbb{Z}_{q}}}[\mathcal{A}(P,g,g^a,g^b,g^{c}) = 1]| < \frac{1}{p(n)}.
\end{equation}
By comparison with Definition \ref{d:PRG}, we therefore see that $\{G_{(p,g,g^a)}\}$ is an indexed collection of pseudorandom generators, under the assumption that $\{\mathrm{QR}_p\}$ satisfies the DDH assumption.
\end{observation}
\noindent We would now like to use the PRG $\{G_{(p,g,g^a)}\}$ to define a suitable classical-secure PRF collection. As $\{G_{(p,g,g^a)}\}$ is a collection of effectively length doubling pseudorandom generators, it has been observed multiple times \cite{naor2004number,bogdanov2017pseudorandom,boneh1998decision} that one could in principle construct a classical-secure PRF collection from $\{G_{(p,g,g^a)}\}$ via the Goldreich-Goldwasser-Micali (GGM) construction~\cite{Goldreich:1986:CRF:6490.6503,goldreich2007foundations}. However, as noted and discussed in Ref.\ \cite{farashahi2007efficient}, one needs to take some care in order to construct such a PRF collection in a rigorous way. More specifically, the GGM construction requires that the functions $G^0_{(p,g,g^a)}$ and $G^1_{(p,g,g^a)}$ (defined via Eq. \eqref{eq:length_doubling}) be iterated in an order defined via the input to the PRF, and for such an iteration to be well defined, and for the GGM proof to hold with only minor modifications, it is essential that there exists an efficient bijection from the codomain of $G^i_{(p,g,g^a)}$ to its domain -- i.e. a function which efficiently enumerates the elements of the group. However, for the group family of quadratic residues, such an efficient bijection exists, and as such it is indeed possible to construct a PRF collection, via the GGM construction, starting from a slightly modified DDH based PRG collection. To make this more precise, given some safe prime $p = 2q+1$, we define the function $f_p:\mathrm{QR}_p\rightarrow \mathbb{Z}_q$ via 
\begin{equation}
    f_p(x) = \begin{cases}
x &\text{ if } x \leq q,\\
p-x &\text{ if } x > q.
\end{cases}
\end{equation}
As noted in Ref.\ \cite{farashahi2007efficient}, this function is fact a bijection, whose inverse $f_p^{-1}:\mathbb{Z}_q \rightarrow \mathrm{QR}_p$ is given by
\begin{equation}
       f_p^{-1}(y) = \begin{cases}
y &\text{ if } y \in \mathrm{QR}_p,\\
p-y &\text{ if } y \notin \mathrm{QR}_p.
\end{cases} 
\end{equation}
While it is clear that $f_p$ can be efficiently computed, efficiency of $f^{-1}_p$ is less obvious, and follows from the fact that group membership in $\mathrm{QR}_p$ can be efficiently tested \cite{blum1984generate}. With this in hand, we now define the indexed collection of functions $\{\tilde{G}_{(p,g,g^a)}\}$, where for all valid parameterizations
\begin{equation}
    \tilde{G}_{(p,g,g^a)}:\mathbb{Z}_q:\rightarrow \mathbb{Z}_q\times\mathbb{Z}_q
\end{equation}
is defined via
\begin{align}
    \tilde{G}_{(p,g,g^a)} &= f_p(G^0_{(p,g,g^a)}(b))||f_p(G^1_{(p,g,g^a)}(b)) \nonumber \\
    &:=\tilde{G}^0_{(p,g,g^a)}(b)||\tilde{G}^1_{(p,g,g^a)}(b).
\end{align}
As $f_p$ is a bijection, we again have, analogously to Observation \ref{o:is_prg}, that $\{\tilde{G}_{(p,g,g^a)}\}$ is an indexed collection of PRG's, under the DDH assumption for $\{\mathrm{QR}_p\}$. Given this, we can finally construct the DDH based PRF which will fulfill all our requirements. Specifically, we consider an indexed collection of keyed functions $\{F_{(p,g,g^a)}\,|\,(p,g,g^a)\in\mathcal{P}_{(p,g,g^a)}\}$, where
\begin{equation}
    F_{(p,g,g^a)}:\mathbb{Z}_q\times\{0,1\}^n \rightarrow \mathbb{Z}_q
\end{equation}
is defined algorithmically, via the GGM construction \cite{Goldreich:1986:CRF:6490.6503,goldreich2007foundations}, as follows:
\begin{algorithm}[H]
  \caption{Algorithmic implementation of $F_{(p,g,g^a)}(b,x)$}\label{alg:new_modified_DDH_PRF}
  \begin{algorithmic}[1]
    \State Given parameterization $p,g,g^a$, as well as key $b\in \mathbb{Z}_q$ and input $x\in \{0,1\}^n$
    \State $b_0 \gets b$
    
    \ForAll{$1 \leq j \leq n$}
        \If{$x_j = 0 $}
        \State $b_j \gets \tilde{G}^0_{(p,g,g^a)}(b_{j-1}) =  f_p(\textsf{modexp}_{p,g}({b_{j-1}})) = f(g^{b_{j-1}} \md(p))$
        \ElsIf{$x_j = 1 $}
        \State $b_{j} \gets \tilde{G}^1_{(p,g,g^a)}(b_{j-1}) =  f_p(\textsf{modexp}_{p,g^a}({b_{j-1}})) = f(g^{ab_{j-1}} \md(p))$
        \EndIf
    \EndFor
    \Statex
    \State \textbf{Output: $b_n \in \mathbb{Z}_q$} \Comment{Note that $b_n = \tilde{G}^{x_n}_{(p,g,g^a)}(\ldots(\tilde{G}^{x_1}_{(p,g,g^a)}(b)))$}
  \end{algorithmic}
\end{algorithm}
\noindent The above algorithm is also illustrated in Fig. \ref{f:GGM}, which serves to illustrate that for all tuples $(b,x)$ the desired output $F_{(p,g,g^a)}(b,x)$ can be calculated by moving through a binary tree, with the key $b\in \mathbb{Z}_q$ at the root, and where at each level either $\tilde{G}^{0}_{(p,g,g^a)}$ or $\tilde{G}^{1}_{(p,g,g^a)}$ is applied, with the path determined by the input string $x\in\{0,1\}^n$. We now make the following claims:

\begin{claim}\label{c:is_PRF}
$\{F_{(p,g,g^a)}\}$ is a classical-secure pseudorandom function collection.
\end{claim}

\begin{claim}\label{c:exists_alg}
For all $n \in \mathbb{N}$, there exists an exact efficient quantum key-learning algorithm which, for all valid parameterizations $(p,g,g^a) \in \mathcal{P}_{n,(p,g,g^a)}$, all $x \in \{0,1\}^n$, and all $b\in \mathbb{Z}_q$, on input $(p,g,g^a),x$ and $F_{(p,g,g^a)}(b,x)$ returns $b$ with probability 1.
\end{claim}
\noindent Before proceeding to prove these claims, we make the following observations: Firstly, note that given Claim \ref{c:is_PRF}, the indexed collection of keyed functions $\{F_{(p,g,g^a)}\}$ satisfies all the requirements of both Theorem \ref{t:kearns} and Corollary \ref{c:kearns}, and as such, for all sufficiently large $n$, the distribution class
\begin{equation}
    \mathcal{C}_n = \{D_{(p,g,g^a),b}\,|\, (p,g,g^a)\in \mathcal{P}_{n,(p,g,g^a)}, b\in \mathbb{Z}_q\},
\end{equation}
defined as per Corollary \ref{c:kearns}, is not efficiently classically PAC learnable with respect to $\sample$ and the TV-distance. Additionally, let us assume that a quantum $\gen$-learner $\mathcal{A}$ is given access to \linebreak $\sample(D_{(p,g,g^a),b})$, for some $(p,g,g^a) \in \mathcal{P}_{n,(p,g,g^a)}$. Note that
\begin{equation}
    \query[\sample(D_{(p,g,g^a),b})] = x||\bin(F_{(p,g,g^a)}(b,x))||\bin(p,g,g^a) \text{ with } x \leftarrow U_n
\end{equation}
and therefore, as illustrated in Fig. \ref{f:proof_sketch}, from any single such sample, the quantum $\gen$-learner $\mathcal{A}$ can run the quantum key-learning algorithm $\mathcal{A}'$ from Claim \ref{c:exists_alg}, on input $(p,g,g^a),x$ and $F_{(p,g,g^a)}(b,x)$, and obtain the key $b$ with probability 1, which coupled with the parameterization $(p,g,g^a)$ provides a complete description of the exact generator $\gen_{((p,g,g^a),b)}$ for $D_{(p,g,g^a),b}$. As a result, it follows from Claim \ref{c:exists_alg} that the distribution class $\{D_{(p,g,g^a),b}\}$ is efficiently quantum PAC $\gen$-learnable with $\sample$ access, and with respect to both the KL-divergence and the TV-distance (due to the fact that the learner in fact always learns an \textit{exact} generator). Given this discussion, we see that the following theorem, providing an affirmative answer to Question~\ref{q:sample_vs_sample}, follows from Claims \ref{c:is_PRF} and \ref{c:exists_alg}:

\begin{figure}
		\centering
		\includegraphics[width=\linewidth]{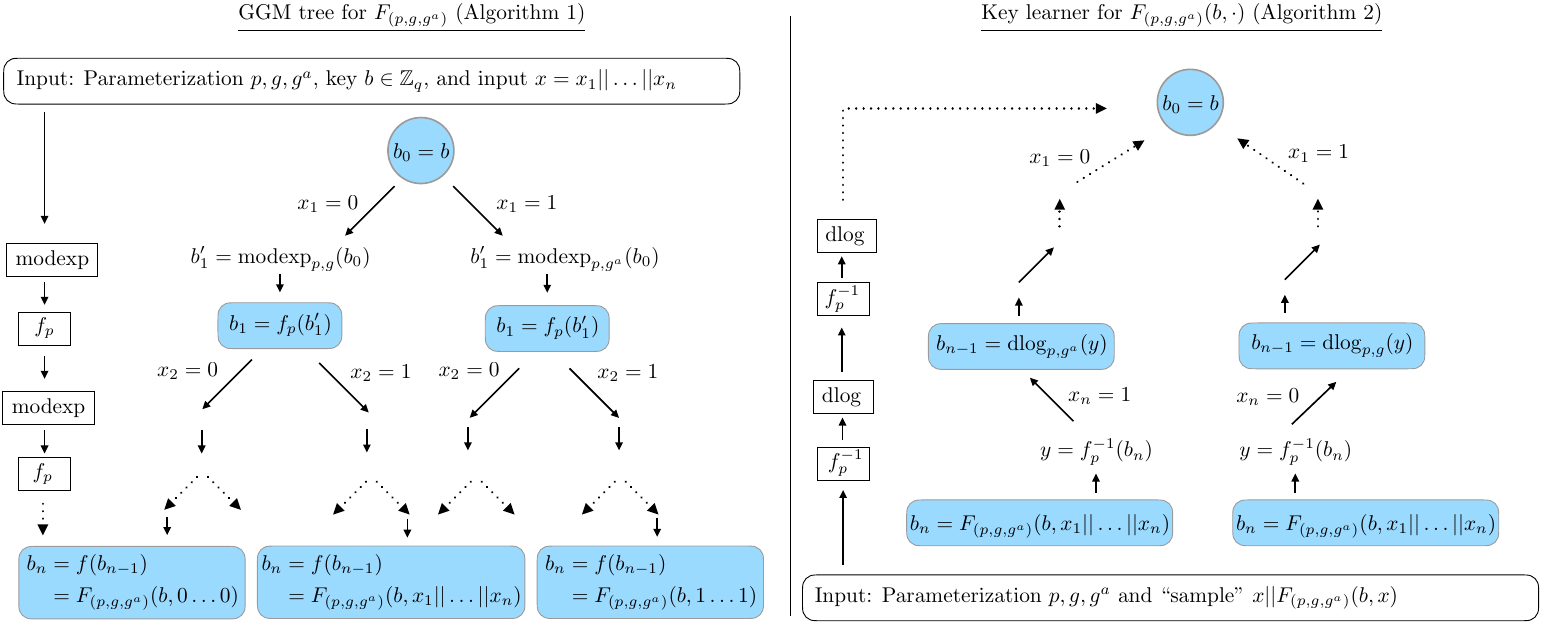}
		\caption{(Left Panel) An illustation of the GGM tree construction for $\{F_{(p,g,g^a)}\}$. For all input tuples $(b,x)$ the desired output $F_{(p,g,g^a)}(b,x)$ can be calculated by moving through a binary tree, with the key $b\in \mathbb{Z}_q$ at the root, and where at each level either $\tilde{G}^{0}_{(p,g,g^a)} = f_p\circ\textsf{modexp}_{p,g}$ or $\tilde{G}^{1}_{(p,g,g^a)} = f_p\circ\textsf{modexp}_{p,g^a}$ is applied, with the path determined by the input string $x\in\{0,1\}^n$. (Right Panel) An illustration of the ``tree-reversal" key-learning algorithm, which given a leaf of the tree $F_{(p,g,g^a)}(b,x)$, along with a description of the path $x$ from the root to this leaf, exploits the knowledge of this path, along with the ability to invert $\tilde{G}^{i}_{(p,g,g^a)}$, to reverse the tree and obtain the value of the root node.}
		\label{f:GGM}
\end{figure}

\begin{theorem}[Quantum $\gen$-Learning Advantage for $\{D_{(p,g,g^a),b}\}$]\label{t:main_result}
For all sufficiently large $n$ the distribution class
\begin{equation}
    \mathcal{C}_n = \{D_{(p,g,g^a),b}\,|\, (p,g,g^a)\in \mathcal{P}_{n,(p,g,g^a)}, b\in \mathbb{Z}_q\}
\end{equation}
is not efficiently classically PAC $\gen$-learnable with respect to the $\sample$ oracle and the TV-distance, but is efficiently quantum PAC $\gen$-learnable with respect to the $\sample$ oracle and the TV-distance.
\end{theorem}
\noindent At this point all that remains is to provide the proofs of Claims \ref{c:is_PRF} and \ref{c:exists_alg}. As $\{\tilde{G}_{(p,g,g^a)}\}$ is a collection of length doubling PRG's the proof of Claim \ref{c:is_PRF} is essentially identical to the original GGM proof \cite{Goldreich:1986:CRF:6490.6503,goldreich2007foundations}, with only minor straightforward modifications required to the definition of the hybrid distributions in order to account for the slightly more general parameterization set. As such, we omit this proof and proceed with a constructive proof of Claim~\ref{c:exists_alg}. Intuitively, as illustrated in Fig. \ref{f:GGM}, given $(p,g,g^a),x$ and $F_{(p,g,g^a)}(b,x)$ -- i.e. a ``leaf" of the GGM tree --  we construct an algorithm which can reverse the path taken through the tree to reach this leaf, and return the value of the root node, which is the desired key. More formally, the proof is as follows:

\begin{proof}[Proof (Claim \ref{c:exists_alg})]
For all tuples $p,g$, where $p = 2q+1$ is an $n$-bit safe prime, and $g$ is a generator for $\mathrm{QR}_p$, we define the discrete log function $\textsf{dlog}_{p,g}:\mathrm{QR}_p \rightarrow \mathbb{Z}_q$ via
\begin{equation}
    \textsf{dlog}_{p,g}(g^x) = x\md(q),
\end{equation}
for all $x\in\mathbb{N}$. As shown in Ref \cite{mosca2004exact}, for all relevant tuples $(p,g)$, there exists an exact efficient quantum algorithm for computing $\textsf{dlog}_{p,g}$ -- i.e. an efficient quantum algorithm which on input $p,g,g^x$ outputs $\textsf{dlog}_{p,g}(g^x)$ with probability 1. Given the ability to efficiently calculate $\textsf{dlog}$ as a subroutine, we can now construct a ``tree-reversal" algorithm by using the string $x$ to determine whether $f_p\circ\textsf{modexp}_{p,g}$ or $f_p\circ\textsf{modexp}_{p,g^a}$ was applied at a given level of the tree, and then applying either $\textsf{dlog}_{p,g}\circ f^{-1}_p$ or $\textsf{dlog}_{p,g^a}\circ f^{-1}_p$ as appropriate. More specifically, the following algorithm, also illustrated in Fig. \ref{f:GGM}, provides a constructive proof for the claim:
\begin{algorithm}[H]
  \caption{Exact efficient quantum key learner for  $F_{p,g,g^a}(b,x)$}\label{alg:key_learner}
  \begin{algorithmic}[1]
    \State Given parameterization $(p,g,g^a)$, as well as a ``sample" $x||F_{(p,g,g^a)}(b,x)$
    \State $b_n \gets F_{(p,g,g^a)}(b,x)$
    
    \ForAll{$n \geq j \geq 1$}
        \State $y \gets f_p^{-1}(b_j)$
        \If{$x_j = 0 $}
        \State $b_{j-1} \gets \mathrm{dlog}_{p,g}(y)$
        \ElsIf{$x_j = 1 $}
        \State $b_{j-1} \gets \mathrm{dlog}_{p,g^a}(y)$
        \EndIf
    \EndFor
    \Statex
    \State \textbf{Output: $b_0 \in \mathbb{Z}_q$}
  \end{algorithmic}
\end{algorithm}
\end{proof}

\subsection{Verification of Quantum $\gen$-Learning Advantage}\label{ss:verification}

In light of our main result --  given in Theorem \ref{t:main_result} -- it is natural to ask whether such a quantum $\gen$-learning advantage could be
efficiently \textit{verified}. 
While there are a variety of ways to formalize the notion of a verification procedure in this context, one approach is as follows: 
We consider an honest referee, Ronald, who for some sufficiently large $n$, is in possession of a full description of some element $D_{(p,g,g^a),b}$ of the distribution class $\mathcal{C}_n$ -- which in this case means in possession of the generator $\gen_{((p,g,g^a),b)}$ as well as its defining parameters $(p,g,g^a) \in \mathcal{P}_{n,(p,g,g^a)}$ and $b\in\mathbb{Z}_q$, and a description of its construction. 
Ronald then selects some valid $(\epsilon,\delta)$ pair and asks Alice, who claims to have a quantum computer, to run the claimed quantum $\gen$-learning algorithm $\mathcal{A}$ for $\mathcal{C}_n$, on input $(\epsilon,\delta)$, by using his access to $\gen_{((p,g,g^a),b)}$ to respond to Alice's $\sample$ queries. 
Alice then obtains from her claimed quantum learning algorithm some classical hypothesis generator $\textsf{HGEN}$, which generates samples from some distribution $D_h$. 
If Alice's learning algorithm is indeed a PAC $\gen$-learner for $\mathcal{C}_n$, then with probability $1-\delta$, $\textsf{HGEN}$ should be a classical $(d_{\mathrm{TV}},\epsilon)$-generator for $D_{(p,g,g^a),b}$. 
The question of verification is then whether there exists some efficient interactive procedure $\mathcal{V}$ (with suitable completeness and soundness guarantees) through which Alice, if $\textsf{HGEN}$ is indeed a $(d_{\mathrm{TV}},\epsilon)$-generator for $D_{(p,g,g^a),b}$, can convince Ronald of this fact.

In the most general (weak-membership) setting of verification Ronald should with high probability reject \textsf{HGEN} if the generated distribution $D_h$ is at least $\epsilon$-far away in TV-distance from the target distribution $D_{((p,g,g^a),b)}$, and accept with high probability if the generated distribution $D_h$ is at least $\epsilon'$-close to the target, for some $\epsilon'< \epsilon$. As such, the weak-membership verification problem, with non-zero completeness $\epsilon'$, is indeed precisely the well studied problem of \textit{tolerant} identity testing, between the learned distribution $D_h$ and the target distribution $D_{((p,g,g^a),b)}$ \cite{canonne2020survey}.

In the simplest (black-box) verification protocol, Alice can only send samples from the generator \textsf{HGEN} to Ronald to convince him of the fact that is indeed a $(d_{\mathrm{TV}},\epsilon)$-generator for $D_{(p,g,g^a),b}$, that is, allow Ronald \emph{random example queries} to $\textsf{HGEN}$. However, given that \textsf{HGEN} is a \textit{classical} generator, Alice could also allow Ronald \emph{membership queries} to $\textsf{HGEN}$ (after publishing a description of the domain). Clearly, such a protocol is stronger as in the latter case Ronald has access not only to random samples from $D_h$, but additional information about the generator as well. Importantly, as alluded to above, we note that membership-query based verification is only a possibility in the case of $\gen$ learners, which output \textit{classical} generators. In the case of $\qgen$ learners, which output \textit{quantum} generators, it is only possible to consider the setting of black-box (sample based) verification.

As we are primarily concerned here with the verification of the $\gen$ learner presented in the previous section, let us focus first on the setting in which Alice publishes the domain of \textsf{HGEN} and allows Ronald membership queries to \textsf{HGEN}. 
Recall that $\gen_{((p,g,g^a),b)}:\{0,1\}^n\rightarrow\{0,1\}^{2n+m}$, i.e. its domain is $\{0,1\}^n$. 
We can now distinguish between two possibilities. 
The first possibility is then that the domain of \textsf{HGEN} is the same as that of $\gen_{((p,g,g^a),b)}$, i.e. $\textsf{HGEN}:\{0,1\}^n\rightarrow \{0,1\}^{2n+m}$,  and the second is that it differs, i.e. $\textsf{HGEN}: \mathcal X \rightarrow \{0,1\}^{2n+m}$, for some domain $\mathcal X \neq \{0,1\}^n$. The following lemma allows us to show that efficient verification is indeed possible in the first case, and therefore in particular, for the $\gen$ learner presented in the previous section.

\begin{lemma}[Generator Distance Upper Bounds TV-Distance] Let $D_1,D_2\in\mathcal{D}_{m}$ be the distributions with classical generator $\gen_{D_1}:\{0,1\}^n\rightarrow\{0,1\}^m$ and $\gen_{D_2}:\{0,1\}^n\rightarrow\{0,1\}^m$ respectively, then
\begin{align}
    d_{\mathrm{TV}}(D_{1},D_{2}) &\leq \mathrm{Pr}_{\substack{x\leftarrow \{0,1\}^n}}\left[\gen_{D_1}(x)\neq \gen_{{D_2}}(x)\right] \label{e:TV_vs_prob}\\
    &:= \mathrm{dist}\left(\gen_{D_1},\gen_{D_2}\right).
\end{align}
\end{lemma}
\begin{proof}
We start with the observation that for any distribution $D\in\mathcal{D}_m$ generated by $\gen_D:\{0,1\}^n\rightarrow\{0,1\}^m$ we have that
\begin{equation}
    D(y) = \frac{1}{2^n}\sum_{x\in\{0,1\}^n}\delta(\gen_D(x),y),
\end{equation}
where $\delta(y,y') =1$ if $y=y'$ and $\delta(y,y') =0$ otherwise. Using this we then have that
\begin{align}
    d_{\mathrm{TV}}(D_{1},D_{2}) &=\frac{1}{2}\sum_{y\in\{0,1\}^m}\left|D_1(y) - D_2(y)\right| \nonumber \\
    &=\frac{1}{2}\sum_{y\in\{0,1\}^m}\left|\frac{1}{2^n}\sum_{x\in\{0,1\}^n}\left(\delta(\gen_{D_1}(x),y) - \delta(\gen_{D_2}(x),y)\right)\right|\nonumber \\
    &\leq\frac{1}{2^{n+1}}\sum_{y\in\{0,1\}^m}\left(\sum_{x\in\{0,1\}^n}\left|\delta(\gen_{D_1}(x),y) - \delta(\gen_{D_2}(x),y)\right|\right) \nonumber \\
    &=\frac{1}{2^{n+1}}\sum_{x\in\{0,1\}^n}\left(\sum_{y\in\{0,1\}^m}\left|\delta(\gen_{D_1}(x),y) - \delta(\gen_{D_2}(x),y)\right|\right)\nonumber \\
    &=\frac{1}{2^{n+1}}\sum_{x\in\{0,1\}^n}\left(2[1-\delta(\gen_{D_1}(x),\gen_{D_2}(x))]\right)\nonumber \\
    &=\frac{1}{2^{n}}\sum_{x\in\{0,1\}^n}[1-\delta(\gen_{D_1}(x),\gen_{D_2}(x))]\nonumber \\
    &=\mathrm{Pr}_{\substack{x\leftarrow \{0,1\}^n}}\left[\gen_{D_1}(x)\neq \gen_{{D_2}}(x)\right].
\end{align}
\end{proof}

In particular, when $\textsf{HGEN}:\{0,1\}^n\rightarrow \{0,1\}^{2n+m}$ we can apply the above lemma, with $\textsf{HGEN}$ as $\gen_{D_1}$ and $\gen_{((p,g,g^a),b)}$ as $\gen_{D_2}$, 
which allows us to see that Ronald can convince himself  that $d_{\mathrm{TV}}(D_h,D_{(p,g,g^a),b)} \leq \epsilon$ by checking that $\mathrm{dist}\left(\textsf{HGEN},\gen_{((p,g,g^a),b)}\right) \leq \epsilon$. To do this, Ronald draws a uniformly random $x \in \{0,1\}^n$, queries both $\textsf{HGEN}$ and $\gen_{((p,g,g^a),b)}$, and outputs $1$ in the event that $\textsf{HGEN}(x) = \gen_{((p,g,g^a),b)}(x)$ and $0$ otherwise.
By repeating this procedure $O(1/\tilde \epsilon^{2} 
\log (1/\eta))$ many times, he can estimate the bias $\mathrm{dist}\left(\textsf{HGEN},\gen_{((p,g,g^a),b)}\right)$ of the corresponding coin up to error $\tilde \epsilon$ with failure probability~$\eta$. 
He accepts the hypothesis generator $\textsf{HGEN}$ if his estimate $b$ satisfies $b + \tilde \epsilon < \epsilon'$ and rejects otherwise. Given this, we note that if Alice runs the quantum $\gen$-learner described in the previous section, then she will obtain, with certainty, the exact generator $\gen_{((p,g,g^a),b)}$ for $D_{(p,g,g^a),b}$, so that $\mathrm{dist}\left(\textsf{HGEN},\gen_{((p,g,g^a),b)}\right) = 0 $, and the above verification procedure can be applied.
In the case of the quantum $\gen$-learner that we have provided, it is therefore straightforward to verify the quantum learning advantage exhibited by this learner.

Finally, if we are in the latter case, i.e. $\textsf{HGEN}: \mathcal X \rightarrow \{0,1\}^{2n+m}$ for some set $\mathcal X \neq \{0,1\}^n$, then it seems unclear how membership queries to the hypothesis generator \textsf{HGEN} would aid Ronald in verifying closeness. 
The problem of verification with membership queries then reduces to the problem of verification with random example queries, that is, of the tolerant identity testing, 
from random samples, of the distributions $D_h$ and $D_{((p,g,g^a),b)}$. 
We note that the distribution $D_{((p,g,g^a),b)}$ is in fact the uniform distribution over a size $2^n$ subset of $\{0,1\}^{2n+m}$, and therefore it is known that there exists a universal constant $\epsilon_0 > 0$ such that for any $\epsilon' \in (0,\epsilon_0)$ and $\epsilon = c\epsilon'$ with $c>1$ at least $\Omega(2^n/(\epsilon' n))$ many samples from $D_h$ are required to verify the hypothesis generator  \cite{valiant2010clt,valiant2011power,canonne2020survey}. In the simpler setting in which $\epsilon'=0$ -- i.e. only generators of the exact target distribution should be accepted while $\epsilon$-false distributions are still rejected  --  it is known that $\Omega(\sqrt{2^n}/\epsilon^2)$ many samples from $D_h$ are required~\cite{canonne2020survey, valiant_automatic_2017,black-box-verification}.

% As the distribution $D_{((p,g,g^a),b)}$ is the uniform distribution over a size $2^n$ subset of $\{0,1\}^{2n+m}$ at least $\Omega(\sqrt{2^n}/\epsilon^2)$ many samples from $D_h$ are required to verify the hypothesis generator \textsf{HGEN} with soundness $\epsilon$ and non-zero completeness $\epsilon' < \epsilon$~\cite{valiant_automatic_2017,black-box-verification}.
% In fact this scaling is also sufficient for the weaker task of identity testing in which $\epsilon' =0$, i.e., only the ideal target distribution is guaranteed to be accepted while $\epsilon$-false distributions are still rejected~ \cite{valiant_automatic_2017,black-box-verification}. 

In summary, we see that the quantum learning advantage exhibited by the learning algorithm provided in the previous section could indeed be efficiently verified, as a result of the fact that the learner always obtains an exact classical generator whose domain coincides with that of the target generator, and therefore allows for efficient membership query based tolerant identity testing. However, in general, verification of a claimed $\gen$-learner (or $\qgen$ learner) reduces to the problem of black-box (sample based) tolerant distribution identity testing, which in general requires exponential sample complexity. Given this, we note that in general, verification of a generative learner is a harder task than verification of a supervised learner, which can be performed efficiently, in the black-box setting \cite{liu2020rigorous}.

\section{On Classical Hardness Results From Alternative Primitives }\label{s:alt_classical_hardness}
One of the key tools which allowed us to prove our main result (Theorem \ref{t:main_result}) was the construction and associated classical hardness result of KMRRSS\ \cite{Kearns:1994:LDD:195058.195155}, presented here (generalized and strengthened) as Theorem \ref{t:kearns}. In this section we shift direction slightly, and explore the possibility of obtaining similar classical hardness results from alternative primitives. More specifically, as we have seen in the previous section, KMRRSS\ have shown that given any classical-secure PRF collection, it is possible to construct a distribution class which is not efficiently classically learnable. In Section \ref{ss:weak_PRFs} we ask whether it is necessary to use a PRF collection to obtain such a classical hardness result, or whether the weaker notion of a \textit{weak}-secure PRF collection \cite{bogdanov2017pseudorandom} would be sufficient. This question is motivated partly by the existence of weak-secure PRF's with known quantum adversaries~\cite{bogdanov2017pseudorandom}, upon which, if one is able to obtain a classical hardness result, then one might be able to obtain additional quantum/classical distribution learning separations, which are possibly achievable by near-term quantum learners. In Section \ref{ss:from_bool_to_dist} we then ask whether one could instantiate the construction of KMRRSS\ with a \textit{Boolean function} concept class which is provably not efficiently classically PAC learnable. In this latter case, the motivation is to understand better the relationship between PAC learning of Boolean functions and PAC learning of discrete distributions, and in particular whether one could leverage existing classical/quantum separations for learning Boolean functions into distribution learning separations. In Sections \ref{ss:weak_PRFs} and \ref{ss:from_bool_to_dist} we formulate conjectures concerning the possibility of using these alternative primitives for classical hardness results, and describe clearly some obstacles towards proving these conjectures.

\subsection{Classical Hardness Results from Weak-Secure Pseudorandom Function Collections}\label{ss:weak_PRFs}

Let us begin by discussing the possibility of proving a classical/quantum learning separation that is based on weak-secure PRFs as opposed to classical-secure PRFs, as they are used in Theorem~\ref{t:kearns}. 
In order to understand the motivation for this question, it is necessary to briefly return to an analysis of the proof of Theorem \ref{t:kearns}. 
In particular, at an informal level (i.e. modulo some details) we recall the following:
\begin{enumerate}
    \item Theorem \ref{t:kearns} states that if $\{F_P\}$ is a classical-secure PRF collection, then for infinitely many $n$ the distribution class $\tilde{\mathcal{C}}_n = \{D_{(P,k)}\}$ is not classically efficiently PAC $\gen$-learnable.
    \item The proof of Theorem \ref{t:kearns} proceeds by assuming that the concept class $\tilde{C}_n$ is efficiently PAC $\gen$-learnable for infinitely many $n$, and then using the associated learning algorithms $\tilde{\mathcal{A}}_n$ to construct a polynomial time algorithm $\mathcal{A}$ which $Q$-infers the function collection $\{F_P\}$, for some polynomial $Q$, thereby contradicting the assumption that $\{F_P\}$ is a classical-secure pseudorandom function collection.
\end{enumerate}
 More specifically,  the proof of Theorem \ref{t:kearns} begins with the observation that
\begin{align}
    \query[\sample(\tilde{D}_{(P,k)})] &= \kgen_{(P,k)}(x) \text { with } x\leftarrow U_n \nonumber \\
    &= x||\bin_n(F_P(k,x)) \text { with } x\leftarrow U_n \nonumber \\
    & = x||\bin_n(\query[\mq(F_P(k,\cdot))](x)) \text { with } x\leftarrow U_n
\end{align}
and therefore when given oracle access to $\mq(F_P(k,\cdot))$, the polynomial inference algorithm $\mathcal{A}$ can simulate the learner $\tilde{\mathcal{A}}_n$ by responding to sample queries of $\tilde{\mathcal{A}}_n$ with $x||\bin_n(\query[\mq(F_P(k,\cdot))](x))$, where $x$ has been drawn uniformly at random. Algorithm $\mathcal{A}$ then uses the obtained generator to ``pass the exam" described in Definition \ref{d:poly_inference}. However, we note that in order to simulate the learner $\tilde{\mathcal{A}}$, it is \textit{not} necessary for the inference algorithm to have membership query access to $F_P(k,\cdot)$. In particular, we have that
\begin{align}
    \query[\sample(\tilde{D}_{(P,k)})] = x||\bin_n(\query[\pex(F_P(k,\cdot),U)])
\end{align}
and therefore the entire proof of Theorem \ref{t:kearns} holds, even if the polynomial inference algorithm $\mathcal{A}$ only has \textit{random example} oracle access to $F_P(k,\cdot)$. In light of this observation, one may immediately think that a \textit{weak}-secure PRF collection would be sufficient for instantiating the construction described in Theorem \ref{t:kearns}. In other words, given that the proof of Theorem~\ref{t:kearns} holds when the polynomial inference algorithm only has random example oracle access to $F_P(k,\cdot)$, it seems plausible that if one can efficiently learn the distribution class $\{\tilde{\mathcal{D}}_{P,k}\}$, then one can use this learner to construct an adversary (i.e. distinguishing algorithm) for the function collection $\{F_P\}$ which only requires random example oracle access. To formalize this idea, we make the following conjecture, a slight variant of Theorem \ref{t:kearns}:

\begin{conjecture}\label{c:weak_prf_construction} Let $\{F_P\}$ be a weak-secure pseudorandom function collection with the additional properties stated in Theorem \ref{t:kearns}. Then, for all sufficiently large $n$ the distribution class $\tilde{\mathcal{C}}_n := \{\tilde{D}_{(P,k)}|P \in \mathcal{P}_n,k\in\mathcal{K}_p\}$ is not efficiently classically PAC $\gen$-learnable with respect to the $\sample$ oracle and the TV-distance.
\end{conjecture}
\noindent Given the proof of Theorem \ref{t:kearns}, and the observation that the polynomial inference algorithm involved in the proof requires only random example oracle access, one might think that a proof of Conjecture \ref{c:weak_prf_construction} would follow immediately. Unfortunately though, this is not the case. However, before discussing the obstacles one faces in adapting the proof of Theorem \ref{t:kearns} to prove Conjecture \ref{t:kearns}, it is worth examining briefly why Conjecture \ref{c:weak_prf_construction} is interesting, and what consequences its truth might have for obtaining classical/quantum distribution learning separations. Firstly, as discussed in Ref.\ \cite{bogdanov2017pseudorandom} there is currently strong evidence that weak-secure PRF's are indeed a less complex object than classical-secure PRF's. More specifically, it is believed that there exist weak-secure PRF's which are not classical-secure PRF's \cite{bogdanov2017pseudorandom}, and as such Conjecture \ref{c:weak_prf_construction} would provide evidence that the existence of classical-secure PRF's is not necessary for the construction of distribution classes which are provably not classically efficiently PAC $\gen$-learnable. Additionally one candidate for such a weak-secure PRF collection is based on the ``learning parity with noise" problem\footnote{We note that this weak-secure PRF collection is in fact a \textit{randomized} function collection -- i.e. the evaluation algorithm is allowed to be probabilistic --  however we leave out a discussion of this subtlety, and refer to Ref.\ \cite{bogdanov2017pseudorandom} for more details.}, which is strongly believed to be hard for classical algorithms with classical random examples \cite{pietrzak2012cryptography}, but which is known to be efficiently solvable by quantum algorithms with \textit{quantum} random examples \cite{Cross_2015, Grilo_2019}. Importantly, unlike quantum algorithms for the discrete logarithm \cite{mosca2004exact}, which seem to require a universal fault-tolerant quantum computer, quantum algorithms for ``learning parity with noise" (LPWN) are robust against certain types of noise models \cite{Cross_2015}, and have in fact already been demonstrated on existing NISQ devices \cite{riste2017demonstration}. As such, while demonstrating a quantum advantage for the generator-learnability of the DDH based concept class described in Section \ref{ss:DDH_result} would require a universal fault-tolerant quantum computer, it is plausible that if Conjecture \ref{c:weak_prf_construction} is true, then one could construct a LPWN based concept class which is not classically efficiently PAC $\gen$-learnable, but which is quantum efficiently PAC $\gen$-learnable using existing or near-term quantum devices, albeit with quantum random samples.

With these observations in mind, let us return to a discussion of the difficulties in adapting the proof of Theorem \ref{t:kearns} into a proof for Conjecture \ref{c:weak_prf_construction}. Analogously to the proof of Theorem \ref{t:kearns}, we would like to show that if the distribution class $\{\tilde{D}_{(P,k)}\}$ is efficiently PAC-generator learnable (with respect to $\sample$ and the KL-divergence) then the function collection $\{F_P\}$ is not weak-secure. However, it is critical to note that in proving Theorem \ref{t:kearns} we relied heavily on the alternative characterization of classical-secure PRF's provided by Lemma \ref{l:gold_predictability}. More specifically, we used the fact that if there exists a polynomial inference algorithm (using membership queries) for an indexed collection of keyed functions, then this collection of functions is not standard-secure. Now, from the previous discussion, we know that if the distribution class $\{\tilde{D}_{(P,k)}\}$ is efficiently PAC $\gen$-learnable, then there exists an efficient polynomial inference algorithm for $\{F_P\}$ which only requires random examples, as opposed to membership queries. However, it is \textit{not} clear that this implies that the function collection $\{F_P\}$ is not weak-secure! In other words, in order to adapt the proof of Theorem \ref{t:kearns} to a proof of Conjecture \ref{c:weak_prf_construction} we need an alternative characterization of weak-secure PRF's analagous to Lemma \ref{l:gold_predictability} -- i.e. a statement that if there exists an efficient polynomial inference algorithm for $\{F_P\}$ which only requires random examples, then $\{F_P\}$ is not a weak-secure PRF collection. To understand why obtaining such a characterization is tricky, it is necessary to sketch the original proof of Lemma \ref{l:gold_predictability} from Ref.\ \cite{Goldreich:1986:CRF:6490.6503}. In order to prove the direction that we are concerned with, one starts by assuming that there exists a polynomial inference algorithm $\mathcal{A}$ for $\{F_P\}$, and then using this algorithm to construct a new distinguishing algorithm $\mathcal{A}'$ which, when given membership query access to some unknown function $F$ can with non-negligible probability determine whether this function was drawn uniformly at random from the set of all functions $F:\tilde{D}_P \rightarrow \mathcal{D}'_P$, or uniformly at random from the set of functions $\{F_P(k,\cdot)\,|\, k \in \mathcal{K}_P\}$. More specifically, algorithm $\mathcal{A}'$ works as follows:

\begin{enumerate}
    \item When given some parameterization $P$, along with oracle access to $\mq(F)$, the distinguishing algorithm $\mathcal{A}'$ begins by simulating the inference algorithm $\mathcal{A}$, which returns an ``exam string" $x$.
    \item Using $\mq(F)$, algorithm $\mathcal{A}'$ then ``prepares the exam" -- i.e. presents algorithm $\mathcal{A}$ with $y_1 = F(x)$ and $y_2 \leftarrow U_{\mathcal{D}'_P}$ in a random order.
    \item The inference algorithm $\mathcal{A}$ then ``takes the exam", and picks either $y_1$ or $y_2$.
    \item If $\mathcal{A}$ picks $y_1$, then $\mathcal{A}'$ outputs 1, otherwise $\mathcal{A}'$ outputs $0$.
\end{enumerate}
The fact that $\mathcal{A}$ succeeds with non-neglible probability then follows straightforwardly from the fact that $\mathcal{A}'$ $Q$-infers $\{F_P\}$ for some polynomial $Q$ \cite{Goldreich:1986:CRF:6490.6503}. In light of the above sketch, we can now analyze the difficulties one faces in adapting the above proof when both the distinguishing algorithm and the inference algorithm are only allowed random example oracle access. Recall, we want to show that if $\mathcal{A}$ $Q$-infers $\{F_P\}$ using only random examples -- i.e. with $\pex$ access -- then we can build a suitable distinguishing algorithm $\mathcal{A}'$, which also only requires random examples (which would imply $\{F_P\}$ is not weak-secure). So, as per the proof sketched above for the case of membership queries, when given access to $\pex(F,U)$, the distinguishing algorithm $\mathcal{A}'$ could start by simulating the inference algorithm $\mathcal{A}$, which returns an exam string $x$. It is at this point that we encounter a problem! Specifically, given the exam string $x$, $\mathcal{A}'$ should prepare the exam by returning $y_1 = F(x)$ along with some $y_2$ drawn uniformly at random from $\mathcal{D}'_P$. If $\mathcal{A}'$ has access to $\mq(x)$ then it is straightforward to prepare the exam, as $\query[\mq(F)](x) = F(x)$. However, if $\mathcal{A}'$ only has access to $\pex(F,U)$, then it \textit{cannot} prepare the exam! Note that if we modified the definition of polynomial inference (given as Definition \ref{d:poly_inference}) so that the inference algorithm does not get to choose its exam string, but is just given an exam string sampled uniformly at random (from the set of strings which have not yet been used), then algorithm $\mathcal{A}'$ \textit{could} prepare the exam for algorithm $\mathcal{A}$, and the rest of the proof would hold, yielding an alternative characterization of weak-secure PRF collections in terms of a slightly modified notion of polynomial inference. However, we note that with such a modified definition of polynomial inference, the proof of Theorem \ref{t:kearns} will no longer work! In particular, recall that the proof of Theorem \ref{t:kearns} relies heavily on the fact that the constructed inference algorithm can use the generator it obtained from the distribution learner to \textit{choose} its own exam string. In other words, if a polynomial inference algorithm for $\{F_P\}$ is required to pass a randomly drawn exam with non-negligible probability, then it is completely unclear how a distribution learner for $\{D_{(P,k)}\}$ can be used to construct a successful polynomial inference algorithm. Given these observations we see that, while Conjecture \ref{c:weak_prf_construction} seems plausible, and has a variety of interesting consequences if true, one cannot simply adapt the proof of Theorem \ref{t:kearns} to this modified setting.

\subsection{Classical Hardness Results from Hard to Learn Function Concept Classes}\label{ss:from_bool_to_dist}

While in this work we have so far focused primarily on the PAC learnability of \textit{distribution} concept classes, as an abstraction of generative modelling, there exists already a large body of work concerning the quantum versus classical PAC learnability of \textit{Boolean function} concept classes \cite{arunachalam2017survey}. In this section, we aim to explore to some extent the relationship between these two notions, and in particular whether existing results in the latter context can be leveraged to obtain results in the former. As a starting point, we note that in principle one could instantiate the distribution class construction from KMRRSS\ \cite{Kearns:1994:LDD:195058.195155} with a Boolean function concept class, as formalized by the following definition:

\begin{definition}
Given some Boolean function $f\in \mathcal{F}_n$, we define the distribution $D_f \in \mathcal{D}_{n+1}$ as the distribution defined via the classical generator $\gen_{D_f}:\{0,1\}^n \rightarrow \{0,1\}^{n+1}$ built from $f$ via
\begin{equation}
    \gen_{D_f}(x) = x||f(x).
\end{equation}
Additionally, given some concept class $\mathcal{C} \subseteq \mathcal{F}_n$ we define the distribution class $\mathcal{D}_\mathcal{C} \subseteq \mathcal{D}_{n+1}$ via
\begin{equation}
    \mathcal{D}_\mathcal{C} = \{D_c\, |c \in \mathcal{C}\}.
\end{equation}
\end{definition}
\noindent Given the above construction, we proceed in this section to prove Theorem \ref{t:learnable_implies_learnable}, and to discuss in detail its inverse statement, which we formalize as Conjecture \ref{c:hard_implies_hard}.

\begin{theorem}[Function Learnability Implies Distribution Learnability]\label{t:learnable_implies_learnable}
If a concept class $\mathcal{C}$ is efficiently classically (quantum) PAC learnable with respect to the uniform distribution and the $\pex$ oracle, then the distribution class $\mathcal{D}_\mathcal{C}$ is efficiently classically (quantum) PAC $\gen$-learnable with respect to the $\sample$ oracle and the TV-distance.
\end{theorem}

\begin{conjecture}[Function Hardness Implies Distribution Hardness]\label{c:hard_implies_hard}
If a concept class $\mathcal{C}$ is not efficiently classically (quantum) PAC learnable with respect to the uniform distribution and the $\pex$ oracle, then the distribution class $\mathcal{D}_\mathcal{C}$ is not efficiently classically (quantum) PAC $\gen$-learnable with respect to the $\sample$ oracle and the TV-distance.
\end{conjecture}
\noindent Apart from shedding some light on the relationship between function learnability and distribution learnability, what we might hope is that taken together Theorem \ref{t:learnable_implies_learnable} and Conjecture \ref{c:hard_implies_hard} (if true) would allow us to instantly leverage some existing separation between the classical versus quantum learnability of a particular Boolean function concept class $\mathcal{C}$, to obtain a separation between the classical versus quantum learnability of the associated distribution class $\mathcal{D}_\mathcal{C}$. Unfortunately however this is not the case. In particular, we stress that both Theorem \ref{t:learnable_implies_learnable} and Conjecture \ref{c:hard_implies_hard} describe a relationship between the generator-learnability of $\mathcal{D}_\mathcal{C}$, and the \textit{distribution specific} PAC learnability of the concept class $\mathcal{C}$, \textit{with respect to the uniform distribution}, as well as with respect to the \textit{classical} random example $\pex$ oracle. More specifically, what this means is that, if Conjecture \ref{c:hard_implies_hard} is true, and if there exists a concept class $\mathcal{C}$ which has the following properties:

\begin{enumerate}[label=(\alph*)]
    \item $\mathcal{C}$ \textit{is not} efficiently \textit{classically} PAC learnable, \textit{with respect to the uniform distribution} and the $\pex$ oracle,
    \item $\mathcal{C}$ \textit{is}  efficiently \textit{quantum} learnable, \textit{with respect to the uniform distribution} and the $\pex$ oracle,
\end{enumerate}
then the distribution class $\mathcal{D}_\mathcal{C}$ would not be efficiently classically PAC $\gen$-learnable (via Conjecture \ref{c:hard_implies_hard}), but it would be efficiently quantum PAC $\gen$-learnable (via Theorem \ref{t:learnable_implies_learnable}). However, at present it is not known whether a concept class $\mathcal{C}$ with both of the above properties exists. More specifically, as discussed in Ref.\ \cite{arunachalam2017survey}, Kearns and Valiant \cite{kearns1994cryptographic} have constructed a concept class which, under the assumption that there exists no efficient algorithm for the factorization of Blum integers, is not efficiently PAC learnable with respect to the $\pex$ oracle, but which Servedio and Gortler \cite{servedio2004equivalences} have shown is efficiently quantum PAC learnable with respect to the $\pex$ oracle. However, recall that in order to prove that a concept class is not efficiently PAC learnable, all one has to do is prove that there exists a single distribution $D$ with respect to which the concept class is not efficiently PAC learnable. As such, it can be that a concept class is not efficiently PAC learnable, while still being efficiently PAC learnable \textit{with respect to the uniform distribution} -- which is the case for the factoring based concept class of Kearns and Valiant.

If one restricts themselves to PAC learnability with respect to the uniform distribution, Bshouty and Jackson \cite{bshouty1998learning} have shown that the concept class of $s$-term DNF, whose best known classical learner with $\pex$ access requires quasi-polynomial time \cite{verbeurgt1990learning}, is efficiently quantum PAC learnable, if one allows the learner access to the \textit{quantum} random example oracle $\qpex$. As such we see that the factoring based concept class of Kearns and Valiant fails to satisfy our requirements due to the fact that it is efficiently classically PAC learnable with respect to the uniform distribution, while the concept class of $s$-term DNF fails to satisfy our requirements due to the fact that the efficient quantum learner requires quantum random examples. Despite this we note that Kharitonov \cite{kharitonov1993cryptographic,kharitonov1995cryptographic} has given a variety of concept classes, which under various  cryptographic assumptions, satisfy property (a) above -- i.e. are not efficiently classically PAC learnable with respect to the uniform distribution and the random example oracle\footnote{We note that Kharitonov's results \cite{kharitonov1993cryptographic,kharitonov1995cryptographic} are in fact significantly stronger. In particular, he provides concept classes which are not even \textit{weakly} learnable (i.e. with non-negligible advantage) with respect to the uniform distribution, even if the learner is allowed membership queries.}. In light of these results, we see that the truth of Conjecture \ref{c:hard_implies_hard} would at the least imply the existence of a distribution class which is not efficiently classically PAC $\gen$-learnable. Given these observations, we proceed to prove Theorem \ref{t:learnable_implies_learnable}, and to discuss in more detail Conjecture \ref{c:hard_implies_hard}.

In order to prove Theorem \ref{t:learnable_implies_learnable} we begin with a few preliminary results and observations. The first such observation follows directly from Definition \ref{d:eff_gen}:

\begin{observation} Let $\gen_D:\{0,1\}^m \rightarrow \{0,1\}^n$ be a classical generator for some probability distribution $D\in\mathcal{D}_n$. Then, for all $y \in \{0,1\}^n$ we have that 
\begin{equation}\label{e:frequency_encoding}
    D(y) = \frac{1}{2^m}\sum_{x\in\{0,1\}^m}\delta(\gen_D(x),y),
\end{equation}
where $\delta(y',y) = 1$ if $y=y'$ and $\delta(y',y) = 0$ otherwise.
\end{observation}
\noindent The above observation then allows us to prove the following lemma:
\begin{lemma}\label{l:inter}
For any two Boolean functions $h,c \in \mathcal{F}_n$, we have that
\begin{equation}
    \mathrm{Pr}_{x\leftarrow U_n}[h(x) \neq c(x)]  =  d_{\mathrm{TV}}(D_h,D_c) 
\end{equation}
\end{lemma}
\begin{proof}
Firstly, note that it follows from Eq. \eqref{e:frequency_encoding} that for any Boolean function $f \in \mathcal{F}_n$, we have for all $y \in \{0,1\}^{n+1}$ that
\begin{align}
    D_f(y) &= \frac{1}{2^n}\sum_{x \in \{0,1\}^n}\delta(\gen_{D_f}(x),y) \nonumber \\
    & = \frac{1}{2^n}\sum_{x \in \{0,1\}^n}\delta\Big(x||f(x),y_{[1,n]}||y_{n+1}\Big) \nonumber \\
    & = \frac{1}{2^n}\delta\Big(f(y_{[1,n]}),y_{n+1}\Big),
\end{align}
where we have denoted the $n$ bit prefix of $y$ with $y_{[1,n]}$ and the $(n+1)$'th bit of $y$ with $y_{n+1}$. Given this, we then have that
\begin{align}
    d_{\mathrm{TV}}(D_h,D_c) &= \frac{1}{2}\sum_{y\in \{0,1\}^{n+1}}|D_h(y) - D_c(y)| \nonumber \\
    & = \frac{1}{2}\left(\frac{1}{2^n}\sum_{y\in \{0,1\}^{n+1}}\Big|\delta\Big(h(y_{[1,n]}),y_{n+1}\Big) - \delta\Big(c(y_{[1,n]}),y_{n+1}\Big)\Big|\right) \nonumber \\
    & = \frac{1}{2}\left(\frac{1}{2^n}\sum_{y \in \{0,1\}^n} 2[1 - \delta(h(y),c(y))] \right)\nonumber \\
    & = \mathrm{Pr}_{y\leftarrow U_n}[h(y) \neq c(y)].
\end{align}
\end{proof}
\noindent Given this, the proof of Theorem \ref{t:learnable_implies_learnable} is then as follows:
\begin{proof}[Proof (Theorem \ref{t:learnable_implies_learnable})]
We consider some concept class $\mathcal{C} \subseteq \mathcal{F}_n$ and begin by noting that for all $c\in \mathcal{C}$
\begin{align}
    &\query(\pex(c,U_n)) = (x,c(x)) \text{ with } x \leftarrow U_n, \\
    &\query(\sample(D_c)) = x||c(x) \text{ with } x \leftarrow U_n.
\end{align}
As such it is clear that any algorithm given oracle access to $\pex(c,U)$ can efficiently simulate oracle access to $\sample(D_c)$, and vice versa. Now, we assume that $\mathcal{C}$ is efficiently classically (quantum) learnable with respect to the uniform distribution and the $\pex$ oracle -- i.e. for all valid $\epsilon$ and $\delta$ there exists an efficient classical (quantum) $(\epsilon,\delta,\pex,U)$-PAC learner for $\mathcal{C}$, which we denote $\mathcal{A}_{\epsilon,\delta}$. Using this we show that for all valid $\epsilon,\delta$ there exists an efficient classical (quantum) $(\epsilon, \delta,\sample,\mathrm{TV})$ PAC $\gen$-learner for the distribution class $\mathcal{D}_\mathcal{C}$, which we denote $\mathcal{A}'_{\epsilon,\delta}$. More specifically, given some valid $\epsilon,\delta$, when given access to $\sample(D_c)$ algorithm $\mathcal{A}'_{\epsilon, \delta}$ does the following:

\begin{enumerate}
    \item Simulate $\mathcal{A}_{\epsilon,\delta}$ on input and obtain some $h\in \mathcal{F}_n$.
    \item Output $\gen_{D_h}$.
\end{enumerate}
By Lemma \ref{l:inter} we know that if $\mathrm{Pr}_{x\leftarrow U_n}[h(x) \neq c(x)] \leq \epsilon$, then $\gen_{D_h}$ is a $(d_{\mathrm{TV}},\epsilon)$ generator for $D_c$. Therefore it follows from the fact that $\mathcal{A}_{\epsilon,\delta}$ is an efficient $(\epsilon,\delta,\pex,U)$-PAC learner for $\mathcal{C}$, that $\mathcal{A}'$ is an $(\epsilon, \delta,\sample,d_{\mathrm{TV}})$-PAC $\gen$-learner for the distribution class $\mathcal{D}_\mathcal{C}$.
\end{proof}
\noindent Before continuing we note that the proof of Theorem \ref{t:learnable_implies_learnable} relies strongly on the fact that the concept class $\mathcal{C}$ is learnable from random examples drawn from the uniform distribution. In particular, if the concept class $\mathcal{C}$ was only learnable with respect to membership queries, or random examples drawn from some other distribution, then the distribution class learner $\mathcal{A}'$ could \textit{not} simulate the concept class learner $\mathcal{A}$. It is this observation that motivates our restriction to the uniform distribution specific learnability of concept classes from random examples. 

Given the above result, we now move onto a discussion of Conjecture \ref{c:hard_implies_hard}. As per the previous discussion, if Conjecture \ref{c:hard_implies_hard} is true, this would allow one to use any concept class which is not efficiently classically PAC learnable, \textit{with respect to the uniform distribution} and random examples (such as those discussed by Kharitonov \cite{kharitonov1993cryptographic,kharitonov1995cryptographic}), to obtain a distribution class which is not efficiently classically PAC $\gen$-learnable with respect to the $\sample$ oracle and the total variation distance. As we will see, a primary obstacle in trying to proving Conjecture \ref{c:hard_implies_hard} is the non-uniqueness of exact classical generators for a given discrete probability distribution. In fact, this difficulty illustrates clearly a key difference between the learnability of Boolean functions and the $\gen$-learnability of distribution classes. More specifically, given a concept class $\mathcal{C}$ which (up to some assumption) is provably not efficiently classically learnable (with respect to random examples drawn from the uniform distribution) a natural proof strategy for Conjecture \ref{c:hard_implies_hard} would be to obtain a contradiction by showing that if the distribution class $\mathcal{D}_\mathcal{C}$ was efficiently $\gen$-learnable, then the concept class $\mathcal{C}$ would be efficiently learnable. Similarly to the proof of Theorem \ref{t:learnable_implies_learnable}, when given access to $\pex(U,c)$ for some $c\in \mathcal{C}$, a function learner $\mathcal{A}$ for $\mathcal{C}$ could easily simulate a distribution-class learner $\mathcal{A}'$ for $\mathcal{D}_\mathcal{C}$ by using the $\pex(U,c)$ oracle to simulate the $\sample(D_c)$ oracle. However, unlike in the proof of Theorem \ref{t:learnable_implies_learnable}, the concept class learner $\mathcal{A}$ cannot simply extract a function hypothesis $h$ from the approximate generator output by $\mathcal{A}'$. To make this more precise, and to pinpoint clearly the key difficulty, we begin with the following series of observations and lemmas which fully characterize the non-uniqueness of \textit{exact} classical generators for $D_c$.

\begin{observation}\label{o:restriction_obs}
For all $m \geq n$ the generator $\gen_{(D_c,m)}:\{0,1\}^m \rightarrow \{0,1\}^{n+1}$ defined via
\begin{equation}
    \gen_{(D_c,m)}(x) := \gen_{D_c}(x_{[1,n]}),
\end{equation}
is an exact generator for $D_c$, where $x := x_{[1,n]}||x_{[n+1,m]}.$
\end{observation}
\begin{proof}
Let $D$ be the distribution generated by $\gen_{(D_C,m)}$. Then, for all $y\in\{0,1\}^n$ we have that
\begin{align}
    D(y||c(y)) &= \frac{1}{2^m}\sum_{x\in\{0,1\}^m}\delta(\gen_{(D_c,m)}(x),y||c(y)) \nonumber\\
    &= \frac{1}{2^m}\sum_{x\in\{0,1\}^m}\delta(\gen_{D_c}(x_{[1,n]}),\gen_{D_c}(y)) \nonumber \\
    &= \frac{1}{2^m}\sum_{x \in \{0,1\}^m}\delta(x_{[1,n]},y) \nonumber \\
    & = \frac{1}{2^m}2^{m-n} \nonumber \\
    &= \frac{1}{2^n}.
\end{align}
\end{proof}

\begin{observation}\label{o:permutation_obs}
For all $m\geq n$, and for all permutations $P:\{0,1\}^m\rightarrow \{0,1\}^m$ the generator $\gen_{(D_c,m,P)}:\{0,1\}^m \rightarrow \{0,1\}^{n+1}$, defined via
\begin{equation}
    \gen_{(D_c,m,P)}(x) := \gen_{(D_c,m)}(P(x)),
\end{equation}
is an exact generator for $D_c$.
\end{observation}
\begin{proof}
As $P$ is a permutation it maps uniformly random inputs to uniformly random outputs, and as such the distribution over output strings of the generator is unaffected by composition with a permutation.
\end{proof}
\noindent Next, we note that we can rule out the possibility of an exact generator with $m < n$:

\begin{lemma}\label{l:error_bounds}
Let $\gen_{D}:\{0,1\}^m \rightarrow \{0,1\}^{n+1}$ be an exact classical generator for some distribution $D \in \mathcal{D}_{n+1}$. Then
\begin{equation}
    d_{\mathrm{TV}}(D,D_c) \geq 1 - 2^{m-n}.
\end{equation}
\end{lemma}
\begin{proof}
We recall that if $\gen_D:\{0,1\}^m \rightarrow \{0,1\}^{n+1}$ then for all $y\in \{0,1\}^{n+1}$
\begin{equation}
    D(y) = \frac{1}{2^m}\sum_{x\in\{0,1\}^m}\delta(\gen_D(x),y).
\end{equation}
From the above it follows that there exist at most $2^m$ strings $y\in\{0,1\}^{n+1}$ with non-zero probability, and all such strings have probability $\alpha/2^m$ for some $\alpha \in \{1,\ldots,2^m\}$. Next, we note that
\begin{equation}
    D_c(y) = \begin{cases} \frac{1}{2^n} \text{ if } y = x||c(x) \text{ for some } x\in \{0,1\}^n, \\
    0 \text { otherwise.}
    \end{cases}
\end{equation}
As such, we see that for $m \leq n$, the optimal $\gen_{D}:\{0,1\}^m \rightarrow \{0,1\}^{n+1}$ which minimizes $d_{\mathrm{TV}}(D,D_c)$ is the one which assigns a probability of $1/2^m$ to $2^m$ of the $2^n$ strings $y \in \{0,1\}^{n+1}$ which are of the form $x||c(x)$ for some $x\in \{0,1\}^n$, and a zero probability to the remaining $2^n - 2^m$ strings. For this optimal generator, we have that
\begin{align}
    d_{\mathrm{TV}}(D,D_c) &= \frac{1}{2}\sum_{y\in\{0,1\}^{n+1}}|D(y) - D_c(y)| \nonumber \\
    &= \frac{1}{2}\sum_{x\in\{0,1\}^{n}}|D(x||c(x)) - \frac{1}{2^n}| \nonumber \\
    &=\frac{1}{2}\left(2^m\left(\frac{1}{2^m} - \frac{1}{2^n} \right) + (2^n - 2^m)\frac{1}{2^n}\right) \nonumber \\
    &= 1 - 2^{m-n}.
\end{align}
The statement follows from the fact that the above was calculated for the optimal generator.
\end{proof}
\noindent An immediate corollary of Lemma \ref{l:error_bounds}, ruling out the possibility of an exact generator for $D_c$ with $m < n$, is then as follows:
\begin{corollary}\label{c:m_restriction}
Given some $\gen_{D}:\{0,1\}^m \rightarrow \{0,1\}^{n+1}$, if $d_{\mathrm{TV}}(D,D_c) < 1/2$, then $m \geq n$.
\end{corollary}
\noindent Finally, given the above results, the following lemma allows us to fully characterize all possible classical exact generators for $D_c$.

\begin{lemma}\label{l:exact_form}
 $\gen_D:\{0,1\}^m \rightarrow \{0,1\}^{n+1}$ is an exact classical generator for $D_c$ if and only if $m\geq n$ and $\gen_{D} = \gen_{(D_c,m,P)}$ for  some permutation $P:\{0,1\}^m\rightarrow\{0,1\}^m$.
\end{lemma}

\begin{proof}
The one direction of the above statement is precisely the content of observations \ref{o:restriction_obs} and \ref{o:permutation_obs}. In the other direction, the fact that $m \geq n$ follows directly from Corollary \ref{c:m_restriction}. Then, in order for $\gen_D$ to be an exact classical generator for $D_c$, we know that for all $y\in \{0,1\}^n$ there are exactly $2^{m-n}$ strings $x\in \{0,1\}^m$ for which 
\begin{equation}
    \gen_{D}(x) = y||c(y) = \gen_{D_c}(y).
\end{equation}
Let us denote the set of these strings as 
\begin{equation}
    X_y = \{x\in \{0,1\}^m\,|\, \gen_{D}(x) = \gen_{D_c}(y)\}
\end{equation}
Additionally, let us denote by $\tilde{X}_y$ the set of all strings $x\in \{0,1\}^m$ for which $\gen_{(D_c,m)}(x) = \gen_{D_c}(y)$, i.e.
\begin{align}
    \tilde{X}_y &= \{x\in \{0,1\}^m \,| \, \gen_{(D_c,m)}(x) = \gen_{D_c}(y) \} \nonumber \\
    &= \{x \in \{0,1\}^m\,|\, x_{[1,n]} = y\}.
\end{align}
Now, as $|X_y| = |\tilde{X}_y| = 2^{m-n}$ we can define a permutation $P_y:X_y \rightarrow \tilde{X}_y$. Additionally, note that for all $x \in \{0,1\}^m$, there is exactly one $y \in \{0,1\}^n$ such that $x \in X_y$, and given this, we can define a permutation $P:\{0,1\}^m \rightarrow \{0,1\}^m$ via
\begin{equation}
    P(x) = P_y(x) \text{ if } x \in X_y.
\end{equation}
Using this, we have that
\begin{equation}
    \gen_{(D_c,m,P)}(x) = \gen_{(D_c,m)}(P(x)) =\gen_{D}(x)
\end{equation}
for all $x \in \{0,1\}^m$.
\end{proof}
\noindent Given the above characterization, let us return to a discussion of one natural strategy to prove Conjecture \ref{c:hard_implies_hard}, and the fundamental obstacle one faces with this strategy. As mentioned before, given a concept class $\mathcal{C}$ which is provably hard to learn, a natural strategy would be to assume that the distribution class $\mathcal{D}_\mathcal{C}$ \textit{is} efficiently PAC $\gen$-learnable, and use this assumption to show that the concept class $\mathcal{C}$ is efficiently PAC learnable, which would be a contradiction. To simplify the exposition, let us make the stronger assumption that the concept class $\mathcal{D}_\mathcal{C}$ is \textit{exactly} $\gen$-learnable -- i.e. for all $\delta$ there exists a $\gen$-learner $\mathcal{A}'_\delta$ which when given access to $\sample(D_c)$ outputs with probability $1-\delta$ an \textit{exact} generator $\gen_D$ for $D_c$. Given this assumption, we want to construct an efficient $(\epsilon,\delta,\pex,U)$-PAC $\gen$-learner $\mathcal{A}_{\epsilon,\delta}$ for the concept class $\mathcal{C}$. A natural approach would be as follows:

\begin{enumerate}
    \item When given access to $\pex(U,c)$, the learner $\mathcal{A}_{\epsilon,\delta}$ simulates $\mathcal{A}'_\delta$ and obtains some generator $\gen_D$, which with probability $1-\delta$ is an exact generator for $D_c$.
    \item Using $\gen_D$, the learner $\mathcal{A}'_{\epsilon,\delta}$ outputs some hypothesis $h \in \mathcal{F}_n$.
\end{enumerate}
Now, as per Lemma \ref{l:exact_form}, we know that if $\gen_D$ is an exact generator for $D_c$, then $\gen_{D} = \gen_{(D_c,m,P)}$ for  some $m\geq n$ and some permutation $P:\{0,1\}^m\rightarrow\{0,1\}^m$. Using this information, and given an exact generator for $D_c$, how should the learner construct its output hypothesis $h$? Well, note that for all $y\in \{0,1\}^m$ for which $y_{[1,n]} = x$ we have that
\begin{align}
\left[\gen_{(D_c,m,P)}(P^{-1}(y))\right]_{[n+1]} &= \left[\gen_{(D_c,m)}(P(P^{-1}(y)))\right]_{[n+1]} \nonumber \\
&= \left[\gen_{(D_c,m)}(y)\right]_{[n+1]}  \nonumber \\
&= \left[\gen_{(D_c)}(y_{[1,n]})\right]_{[n+1]} \nonumber \\
&= \left[\gen_{(D_c)}(x)\right]_{[n+1]} \nonumber \\
&= \left[x||c(x)\right]_{[n+1]}  \nonumber\\
&= c(x) .
\end{align}
As such, if the learner $\mathcal{A}$ knew or could learn the permutation $P^{-1}$, then it could simply output the hypothesis $h\in \mathcal{F}_n$ defined via
\begin{equation}
    h(x) = \left[\gen_{(D_c,m,P)}(P^{-1}(x||0\ldots 0))\right]_{[n+1]} = c(x)
\end{equation}
which would in fact be an exact solution. The key point, however, is that the learner $\mathcal{A}$ \textit{does not} even know the permutation $P$. A natural question is then whether $\mathcal{A}$ could use $\gen_{(D_c,m,P)}$ to \textit{learn} $P^{-1}$? Well, we note that
\begin{equation}
    \gen_{(D_c,m,P)}(x) = P(x)_{[1,n]}||c(P(x)_{[1,n]})
\end{equation}
and so, at least in the case that $m = n$, it is clear that one can generate a data-set of input/output pairs $(x,P(x)) := (P^{-1}(y),y)$. Unfortunately, 
however, it is known that even with respect to membership queries, there does not exist an efficient \emph{exact} learner for the concept class of permutations \cite{arvind1996complexity}, and so the possibility of efficiently \textit{exactly} learning $P^{-1}$ from $\gen_{(D_c,m,P)}$ is ruled out. Of course, in principle it could be sufficient to learn an approximation to $P^{-1}$ from polynomially many random examples, however whether or not this is possible efficiently is not known. Additionally, as mentioned before, all of this is under the overly strong assumption that the generator learner is an exact learner, which outputs an exact generator with $m=n$. As can be seen from the above discussion, lifting either of these assumptions makes the fundamental problem of defining a suitable hypothesis from the output generator significantly harder.

At this stage we have outlined the primary difficulty with one natural strategy for proving Conjecture \ref{c:hard_implies_hard}, which provides a clear illustration of a key conceptual and technical difference between the PAC learnability of Boolean function concept classes and the generator learnability of distribution classes. Of course, one could conceive of a variety of other strategies, based for example on alternative characterizations of efficient PAC learnabilty \cite{haussler1991equivalence}, Occam algorithms \cite{board1992necessity} or VC dimensions \cite{blumer1989learnability}, however it is important to keep in mind the restriction of \textit{efficient learnability with respect to random examples from the uniform distribution}, which makes it unclear how to immediately apply existing results involving some of the above mentioned tools and characterizations.

\section{Discussion and Conclusion}\label{s:conclusion}

Given the results and insights of this work, we provide here a brief summary, as well as a perspective on interesting open questions and future directions. Firstly, to summarize, in Section \ref{s:Seperation} we have constructed a class of probability distributions, specified by \textit{classical} generators, which under the DDH assumption, is provably not efficiently PAC $\gen$-learnable by any classical algorithm, but for which we have constructed an efficient quantum $\gen$-learner, which only requires access to \textit{classical} samples from the unknown distribution. This construction therefore provides a clear example of a generative modelling task for which quantum learners exhibit a provable quantum advantage. Despite this, there of course remain a variety of interesting open questions, for which the insights and conjectures from Section \ref{s:alt_classical_hardness} may provide useful:

\begin{enumerate}
    \item What can one say about the quantum versus classical PAC learnability (in a generative sense) of the probability distributions used for the demonstration of ``quantum computational supremacy" \cite{bremner_average-case_2016, BosonSampling, arute2019quantum}. In particular, there are a variety of distinct questions. Firstly, while it is known that there exists no efficient classical algorithm mapping from randomly drawn quantum circuit parameters to samples from a distribution close to the one generated by the corresponding quantum circuit, is it possible to prove that when \textit{given} samples from such distributions (as opposed to circuit parameters) it is also not possible to efficiently classically learn a description of a suitable generator? Intuitively, this question seems closely related to the question of whether or not efficient classical \emph{verification} of such distributions is possible. To this end, we note that Ref.~\cite{black-box-verification} has proven that efficiently verifying certain such distributions given classical samples is not efficiently possible. However, it seems plausible that the existence of an efficient classical PAC $\gen$-learner would imply the existence of an efficient classical black-box verification algorithm -- which would then rule out the possibility of an efficient classical PAC $\gen$-learner. However, there are two important obstacles. Firstly, as discussed at length in Section \ref{ss:from_bool_to_dist}, it is important to note that there is no \textit{unique} generator for a given probability distribution. Secondly, the PAC framework places no requirement on the behaviour of $\gen$-learners when given access to samples from some distribution outside of the learnable distribution class. As such, while plausible, it is not clear how exactly to exploit an efficient PAC $\gen$-learner for the construction of an efficient black-box verification algorithm with suitable soundness guarantees, and formalizing this connection would certainly be of great interest. Secondly, in addition to proving hardness of classically learning such distributions, we would of course also like to investigate the possibility of efficient quantum learnability, and how this relates to quantum verifiability \cite{black-box-verification,GWD16,NewSupremacy,miller_quantum_2017,hayashi_verifying_2019}. Once again, while it is known that there exist efficient quantum generators for such distributions, this certainly does not immediately imply the existence of efficient quantum PAC $\qgen$-learners. As such, understanding whether or not there exist efficient quantum PAC $\qgen$-learners for such distributions is of natural interest, particularly in light of the potential connections between generator-learnability and black-box verification.

    \item As the quantum learner that we have constructed in Section \ref{ss:DDH_result} relies on the quantum algorithm for discrete logarithms~ \cite{mosca2004exact}, it is most likely the case that the quantum advantage exhibited by this learner would require the existence of a universal fault-tolerant quantum computer. Given the current availability of NISQ devices, it is of natural interest to ask whether there exists a generative modelling problem for which near-term quantum learners (and in particular $\qgen$-learners) can exhibit a provable advantage over classical learning algorithms. In order to answer this question it will certainly be necessary to understand better the theoretical properties of previously proposed NISQ hybrid-quantum classical algorithms for generative modelling, such as Born machines \cite{coyle2019born}. Additionally, it also seems likely that techniques for proving classical hardness results under weaker assumptions, as discussed in Section \ref{s:alt_classical_hardness}, would be of great help. Alternatively, it may help to focus on probability distributions which can be generated by near term quantum devices, but not classical devices, as discussed in the previous point. It is also important to reiterate that the seperation we have obtained here relies fundamentally on the known advantage quantum computers offer for computing discrete logarithms, and as such this work does not provide a new primitive for proving classical/quantum separations. Whether one can construct a quantum/classical learning separation without relying on such prior primitives is an interesting open question.

    \item It is of interest to note that the efficient quantum $\gen$-learner that we have constructed in Section \ref{ss:DDH_result} requires only a \textit{single} oracle query, and \textit{always} outputs an \textit{exact} generator. Such a quantum $\gen$-learner is certainly formally sufficient for the purpose of answering Question~\ref{q:sample_vs_sample} in the affirmative, and from one perspective provides the ``optimal" $\gen$-learner, in the sense that its query complexity is clearly optimal, and its behaviour (both run-time and output) is independent of $\epsilon$ and $\delta$ -- i.e. for all $\epsilon$ and $\delta$ the algorithm returns an exact generator with certainty. However, intuitively we might think of a ``learning" algorithm as an algorithm which requires multiple samples (i.e. learns from a ``data-set"), and outputs most often only approximate solutions, and from this perspective it is not clear to which extent the $\gen$-learner we have constructed can be considered a ``learning algorithm". As such, from a conceptual perspective it is interesting to ask whether there exists a distribution class which provides an affirmative answer to Question \ref{q:sample_vs_sample}, but for which the efficient quantum generator learner requires a non-trivial query complexity, and at best outputs only a suitably approximate generator, with sufficiently high probability. In particular, while the $\gen$-learner we have constructed is clearly highly specific to the distribution class we have constructed, it is possible that by considering concept classes for which always exact constant query complexity learners do not exist, one may be forced to construct or consider quantum generative modelling algorithms which are not as task-specific as the learner we have constructed here, and may also be suitable for near-term devices. 
    \item In this work we have considered quantum and classical $\gen$-learners, both of which only have access to \textit{classical} samples from the unknown probability distribution -- i.e. to the $\sample$ oracle. Analogously to the Boolean function setting \cite{arunachalam2017survey}, it is also interesting to ask whether there exists a distribution class for which a quantum learner (either a $\gen$-learner or a $\qgen$-learner) exhibits a quantum advantage, but only if the quantum learner has access to \textit{quantum} samples from the $\qsample$ oracle. As discussed in Section \ref{ss:weak_PRFs}, it seems likely that if Conjecture \ref{c:weak_prf_construction} is true then one could construct such a concept class using the weak-secure pseudorandom function collection based on the Learning Parity with Noise problem. Additionally, it is also plausible that if Conjecture \ref{c:hard_implies_hard} is true, then one could modify both this and Theorem \ref{t:learnable_implies_learnable} to prove both learnability and hardness results for quantum learners with quantum samples, from corresponding results for Boolean function concept classes.
\end{enumerate}

\begin{acknowledgements}
RS is funded by the BMWi, under the PlanQK initiative, and is grateful for many fruitful discussions with Alexander Nietner,  Nathan Walk and the QML reading group at the FU Berlin.
JPS acknowledges funding of the Berlin Institute for the Foundations of Learning and Data, the Einstein Foundation Berlin and the BMBF ``Post-Quantum-Cryptography" framework. JE is funded by the BMWi under the PlanQK initiative and the DFG 
(CRC 183, Project B01, EI 519/14-1, Math$+$ and EF1-11).
DH is funded by the Templeton Foundation and is grateful for discussions with Brian Coyle. The authors also acknowledge helpful and insightful comments from the anonymous referees, and in particular for pointing out that Theorem \ref{t:kearns} could most likely be strengthened to hold for the TV-distance.
\end{acknowledgements}

% \bibliography{literature.bib}

\bibliographystyle{abbrvunsrtnat}
\bibliography{literature.bib}

\begin{thebibliography}{62}
\providecommand{\natexlab}[1]{#1}
\providecommand{\url}[1]{\texttt{#1}}
\expandafter\ifx\csname urlstyle\endcsname\relax
  \providecommand{\doi}[1]{doi: #1}\else
  \providecommand{\doi}{doi: \begingroup \urlstyle{rm}\Url}\fi

\bibitem[Valiant(1984)]{valiant1984theory}
L.~G. Valiant.
\newblock A theory of the learnable.
\newblock \emph{Communications of the ACM}, 27\penalty0 (11):\penalty0
  1134--1142, 1984.
\newblock \doi{10.1145/1968.1972}.

\bibitem[Kearns et~al.(1994{\natexlab{a}})Kearns, Vazirani, and
  Vazirani]{kearns1994introduction}
M.~J. Kearns, U.~V. Vazirani, and U.~Vazirani.
\newblock \emph{An introduction to computational learning theory}.
\newblock MIT press, 1994{\natexlab{a}}.
\newblock \doi{10.7551/mitpress/3897.001.0001}.

\bibitem[Shalev-Shwartz and Ben-David(2014)]{shalev2014understanding}
S.~Shalev-Shwartz and S.~Ben-David.
\newblock \emph{Understanding machine learning: From theory to algorithms}.
\newblock Cambridge University Press, 2014.
\newblock \doi{10.1017/CBO9781107298019}.

\bibitem[Arunachalam and de~Wolf(2017)]{arunachalam2017survey}
S.~Arunachalam and R.~de~Wolf.
\newblock Guest column: A survey of quantum learning theory.
\newblock \emph{ACM SIGACT News}, 48\penalty0 (2):\penalty0 41--67, 2017.
\newblock \doi{10.1145/3106700.3106710}.

\bibitem[Ciliberto et~al.(2020)Ciliberto, Rocchetto, Rudi, and
  Wossnig]{ciliberto2020statistical}
C.~Ciliberto, A.~Rocchetto, A.~Rudi, and L.~Wossnig.
\newblock Statistical limits of supervised quantum learning.
\newblock \emph{Physical Review A}, 102\penalty0 (4), Oct 2020.
\newblock \doi{10.1103/physreva.102.042414}.

\bibitem[Dunjko and Briegel(2018)]{dunjko2018machine}
V.~Dunjko and H.~J. Briegel.
\newblock Machine learning \& artificial intelligence in the quantum domain: a
  review of recent progress.
\newblock \emph{Rep. Prog. Phys.}, 81\penalty0 (7):\penalty0 074001, 2018.
\newblock \doi{10.1088/1361-6633/aab406}.

\bibitem[Schuld and Petruccione(2018)]{schuld2018supervised}
M.~Schuld and F.~Petruccione.
\newblock \emph{Supervised learning with quantum computers}.
\newblock Springer, 2018.
\newblock \doi{10.1007/978-3-319-96424-9}.

\bibitem[Biamonte et~al.(2017)Biamonte, Wittek, Pancotti, Rebentrost, Wiebe,
  and Lloyd]{biamonte2017quantum}
J.~Biamonte, P.~Wittek, N.~Pancotti, P.~Rebentrost, N.~Wiebe, and S.~Lloyd.
\newblock Quantum machine learning.
\newblock \emph{Nature}, 549\penalty0 (7671):\penalty0 195--202, 2017.
\newblock \doi{10.1038/nature23474}.

\bibitem[Bshouty and Jackson(1998)]{bshouty1998learning}
N.~H. Bshouty and J.~C. Jackson.
\newblock {Learning DNF over the uniform distribution using a quantum example
  oracle}.
\newblock \emph{SIAM Journal on Computing}, 28\penalty0 (3):\penalty0
  1136--1153, 1998.
\newblock \doi{10.1145/225298.225312}.

\bibitem[Canonne(2020{\natexlab{a}})]{canonne2020short}
C.~L. Canonne.
\newblock A short note on learning discrete distributions.
\newblock \emph{arXiv preprint arXiv:2002.11457}, 2020{\natexlab{a}}.

\bibitem[Kamath et~al.(2015)Kamath, Orlitsky, Pichapati, and
  Suresh]{kamath2015learning}
S.~Kamath, A.~Orlitsky, D.~Pichapati, and A.~T. Suresh.
\newblock On learning distributions from their samples.
\newblock In P.~Grünwald, E.~Hazan, and S.~Kale, editors, \emph{Proceedings of
  The 28th Conference on Learning Theory}, volume~40 of \emph{Proceedings of
  Machine Learning Research}, pages 1066--1100, Paris, France, 03--06 Jul 2015.
  PMLR.
\newblock URL \url{http://proceedings.mlr.press/v40/Kamath15.html}.

\bibitem[Diakonikolas(2016)]{diakonikolas2016learning}
I.~Diakonikolas.
\newblock Learning structured distributions.
\newblock \emph{Handbook of Big Data}, 267, 2016.
\newblock \doi{10.1201/b19567}.

\bibitem[Goodfellow et~al.(2014)Goodfellow, Pouget-Abadie, Mirza, Xu,
  Warde-Farley, Ozair, Courville, and Bengio]{goodfellow2014generative}
I.~Goodfellow, J.~Pouget-Abadie, M.~Mirza, B.~Xu, D.~Warde-Farley, S.~Ozair,
  A.~Courville, and Y.~Bengio.
\newblock Generative adversarial nets.
\newblock In \emph{Advances in neural information processing systems}, pages
  2672--2680, 2014.

\bibitem[Kingma and Welling(2013)]{kingma2013auto}
D.~P. Kingma and M.~Welling.
\newblock Auto-encoding variational bayes.
\newblock \emph{arXiv preprint arXiv:1312.6114}, 2013.

\bibitem[Kobyzev et~al.(2020)Kobyzev, Prince, and Brubaker]{Kobyzev_2020}
I.~Kobyzev, S.~Prince, and M.~Brubaker.
\newblock Normalizing flows: An introduction and review of current methods.
\newblock \emph{IEEE Transactions on Pattern Analysis and Machine
  Intelligence}, pages 1--16, 2020.
\newblock \doi{10.1109/tpami.2020.2992934}.

\bibitem[Coyle et~al.(2020)Coyle, Mills, Danos, and Kashefi]{coyle2019born}
B.~Coyle, D.~Mills, V.~Danos, and E.~Kashefi.
\newblock {The Born supremacy: quantum advantage and training of an Ising Born
  machine}.
\newblock \emph{npj Quantum Information}, 6:\penalty0 60, 2020.
\newblock \doi{10.1038/s41534-020-00288-9}.

\bibitem[Liu and Wang(2018)]{liu2018differentiable}
J.-G. Liu and L.~Wang.
\newblock {Differentiable learning of quantum circuit Born machines}.
\newblock \emph{Phys. Rev. A}, 98\penalty0 (6):\penalty0 062324, 2018.
\newblock \doi{10.1103/PhysRevA.98.062324}.

\bibitem[Benedetti et~al.(2019)Benedetti, Garcia-Pintos, Perdomo,
  Leyton-Ortega, Nam, and Perdomo-Ortiz]{Benedetti_2019}
M.~Benedetti, D.~Garcia-Pintos, O.~Perdomo, V.~Leyton-Ortega, Y.~Nam, and
  A.~Perdomo-Ortiz.
\newblock A generative modeling approach for benchmarking and training shallow
  quantum circuits.
\newblock \emph{npj Quantum Information}, 5\penalty0 (1), May 2019.
\newblock \doi{10.1038/s41534-019-0157-8}.

\bibitem[Gao et~al.(2018)Gao, Zhang, and Duan]{Gaoeaat9004}
X.~Gao, Z.~Zhang, and L.~Duan.
\newblock A quantum machine learning algorithm based on generative models.
\newblock \emph{Science Advances}, 4\penalty0 (12):\penalty0 eaat9004, 2018.
\newblock \doi{10.1126/sciadv.aat9004}.

\bibitem[Dallaire-Demers and Killoran(2018)]{Dallaire_Demers_2018}
P.-L. Dallaire-Demers and N.~Killoran.
\newblock Quantum generative adversarial networks.
\newblock \emph{Phys. Rev. A}, 98, 2018.
\newblock \doi{10.1103/physreva.98.012324}.

\bibitem[Hu et~al.(2019)Hu, Wu, Cai, Ma, Mu, Xu, Wang, Song, Deng, Zou, and
  Sun]{Hueaav2761}
L.~Hu, S.-H. Wu, W.~Cai, Y.~Ma, X.~Mu, Y.~Xu, H.~Wang, Y.~Song, D.-L. Deng,
  C.-L. Zou, and L.~Sun.
\newblock Quantum generative adversarial learning in a superconducting quantum
  circuit.
\newblock \emph{Science Advances}, 5\penalty0 (1):\penalty0 eaav2761, 2019.
\newblock \doi{10.1126/sciadv.aav2761}.

\bibitem[Lloyd and Weedbrook(2018)]{lloyd2018quantum}
S.~Lloyd and C.~Weedbrook.
\newblock Quantum generative adversarial learning.
\newblock \emph{Phys. Rev. Lett.}, 121\penalty0 (4):\penalty0 040502, 2018.
\newblock \doi{10.1103/PhysRevLett.121.040502}.

\bibitem[Chakrabarti et~al.(2019)Chakrabarti, Yiming, Li, Feizi, and
  Wu]{chakrabarti2019quantum}
S.~Chakrabarti, H.~Yiming, T.~Li, S.~Feizi, and X.~Wu.
\newblock {Quantum Wasserstein generative adversarial networks}.
\newblock In \emph{Advances in Neural Information Processing Systems}, pages
  6781--6792, 2019.

\bibitem[Verdon et~al.(2019)Verdon, Marks, Nanda, Leichenauer, and
  Hidary]{verdon2019quantum}
G.~Verdon, J.~Marks, S.~Nanda, S.~Leichenauer, and J.~Hidary.
\newblock {Quantum Hamiltonian-ased models and the variational quantum
  thermalizer algorithm}, 2019.
\newblock arXiv preprint arXiv:1910.02071.

\bibitem[Aaronson and Arkhipov(2011)]{BosonSampling}
S.~Aaronson and A.~Arkhipov.
\newblock The computational complexity of linear optics.
\newblock In \emph{Proceedings of the Forty-Third Annual ACM Symposium on
  Theory of Computing}, STOC '11, page 333–342, New York, NY, USA, 2011.
  Association for Computing Machinery.
\newblock \doi{10.1145/1993636.1993682}.

\bibitem[Bremner et~al.(2016)Bremner, Montanaro, and
  Shepherd]{bremner_average-case_2016}
M.~J. Bremner, A.~Montanaro, and D.~J. Shepherd.
\newblock Average-case complexity versus approximate simulation of commuting
  quantum computations.
\newblock \emph{Phys. Rev. Lett.}, 117:\penalty0 080501, Aug 2016.
\newblock \doi{10.1103/PhysRevLett.117.080501}.

\bibitem[Arute et~al.(2019)Arute, Arya, Babbush, Bacon, Bardin, Barends,
  Biswas, Boixo, Brandao, Buell, et~al.]{arute2019quantum}
F.~Arute, K.~Arya, R.~Babbush, D.~Bacon, J.~C. Bardin, R.~Barends, R.~Biswas,
  S.~Boixo, F.~G. S.~L. Brandao, D.~A. Buell, et~al.
\newblock Quantum supremacy using a programmable superconducting processor.
\newblock \emph{Nature}, 574\penalty0 (7779):\penalty0 505--510, 2019.
\newblock \doi{10.1038/s41586-019-1666-5}.

\bibitem[Boneh(1998)]{boneh1998decision}
D.~Boneh.
\newblock {The decision Diffie-Hellman problem}.
\newblock In \emph{International Algorithmic Number Theory Symposium}, pages
  48--63. Springer, 1998.
\newblock \doi{10.1007/BFb0054851}.

\bibitem[Kearns et~al.(1994{\natexlab{b}})Kearns, Mansour, Ron, Rubinfeld,
  Schapire, and Sellie]{Kearns:1994:LDD:195058.195155}
M.~Kearns, Y.~Mansour, D.~Ron, R.~Rubinfeld, R.~E. Schapire, and L.~Sellie.
\newblock On the learnability of discrete distributions.
\newblock In \emph{Proceedings of the Twenty-sixth Annual ACM Symposium on
  Theory of Computing}, STOC '94, pages 273--282, New York, NY, USA,
  1994{\natexlab{b}}. ACM.
\newblock \doi{10.1145/195058.195155}.

\bibitem[Mosca and Zalka(2004)]{mosca2004exact}
M.~Mosca and C.~Zalka.
\newblock Exact quantum fourier transforms and discrete logarithm algorithms.
\newblock \emph{International Journal of Quantum Information}, 2\penalty0
  (01):\penalty0 91--100, 2004.
\newblock \doi{10.1142/S0219749904000109}.

\bibitem[Goldreich et~al.(1986)Goldreich, Goldwasser, and
  Micali]{Goldreich:1986:CRF:6490.6503}
O.~Goldreich, S.~Goldwasser, and S.~Micali.
\newblock How to construct random functions.
\newblock \emph{J. ACM}, 33\penalty0 (4):\penalty0 792--807, August 1986.
\newblock \doi{10.1145/6490.6503}.

\bibitem[Lehmann and Romano(2006)]{lehmann2006testing}
E.~L. Lehmann and J.~P. Romano.
\newblock \emph{Testing statistical hypotheses}.
\newblock Springer Science \& Business Media, 2006.
\newblock \doi{10.1111/j.1467-985X.2007.00473_10.x}.

\bibitem[Goldreich(2007)]{goldreich2007foundations}
O.~Goldreich.
\newblock \emph{Foundations of cryptography: volume 1, basic tools}.
\newblock Cambridge University Press, 2007.
\newblock \doi{10.1017/CBO9780511546891}.

\bibitem[Katz and Lindell(2007)]{Katz:2007:IMC:1206501}
J.~Katz and Y.~Lindell.
\newblock \emph{Introduction to Modern Cryptography (Chapman \& Hall/Crc
  Cryptography and Network Security Series)}.
\newblock Chapman \& Hall/CRC, 2007.
\newblock \doi{10.1201/b17668}.

\bibitem[Zhandry(2012)]{Zhandry:2012:CQR:2417500.2417838}
M.~Zhandry.
\newblock How to construct quantum random functions.
\newblock In \emph{Proceedings of the 2012 IEEE 53rd Annual Symposium on
  Foundations of Computer Science}, FOCS '12, pages 679--687. IEEE Computer
  Society, 2012.
\newblock \doi{10.1109/FOCS.2012.37}.

\bibitem[Bogdanov and Rosen(2017)]{bogdanov2017pseudorandom}
A.~Bogdanov and A.~Rosen.
\newblock Pseudorandom functions: Three decades later.
\newblock In \emph{Tutorials on the Foundations of Cryptography}, pages
  79--158. Springer, 2017.
\newblock \doi{10.1007/978-3-319-57048-8_3}.

\bibitem[Blum and Micali(1984)]{blum1984generate}
M.~Blum and S.~Micali.
\newblock How to generate cryptographically strong sequences of pseudorandom
  bits.
\newblock \emph{SIAM journal on Computing}, 13\penalty0 (4):\penalty0 850--864,
  1984.
\newblock \doi{10.1109/SFCS.1982.72}.

\bibitem[Naor and Reingold(2004)]{naor2004number}
M.~Naor and O.~Reingold.
\newblock Number-theoretic constructions of efficient pseudo-random functions.
\newblock \emph{Journal of the ACM (JACM)}, 51\penalty0 (2):\penalty0 231--262,
  2004.
\newblock \doi{10.1145/972639.972643}.

\bibitem[Farashahi et~al.(2007)Farashahi, Schoenmakers, and
  Sidorenko]{farashahi2007efficient}
R.~R. Farashahi, B.~Schoenmakers, and A.~Sidorenko.
\newblock Efficient pseudorandom generators based on the ddh assumption.
\newblock In \emph{International Workshop on Public Key Cryptography}, pages
  426--441. Springer, 2007.
\newblock \doi{10.1007/978-3-540-71677-8_28}.

\bibitem[Canonne(2020{\natexlab{b}})]{canonne2020survey}
C.~L. Canonne.
\newblock \emph{A Survey on Distribution Testing: Your Data is Big. But is it
  Blue?}
\newblock Number~9 in Graduate Surveys. Theory of Computing Library,
  2020{\natexlab{b}}.
\newblock \doi{10.4086/toc.gs.2020.009}.

\bibitem[Valiant and Valiant(2010)]{valiant2010clt}
G.~Valiant and P.~Valiant.
\newblock A clt and tight lower bounds for estimating entropy.
\newblock In \emph{Electronic Colloquium on Computational Complexity (ECCC)},
  volume~17, page~9. Citeseer, 2010.

\bibitem[Valiant and Valiant(2011)]{valiant2011power}
G.~Valiant and P.~Valiant.
\newblock The power of linear estimators.
\newblock In \emph{2011 IEEE 52nd Annual Symposium on Foundations of Computer
  Science}, pages 403--412. IEEE, 2011.
\newblock \doi{10.1109/FOCS.2011.81}.

\bibitem[Valiant and Valiant(2017)]{valiant_automatic_2017}
G.~Valiant and P.~Valiant.
\newblock An automatic inequality prover and instance optimal identity testing.
\newblock \emph{{SIAM} J. Comput.}, 46\penalty0 (1):\penalty0 429--455, 2017.
\newblock \doi{10.1137/151002526}.

\bibitem[Hangleiter et~al.(2019)Hangleiter, Kliesch, Eisert, and
  Gogolin]{black-box-verification}
D.~Hangleiter, M.~Kliesch, J.~Eisert, and C.~Gogolin.
\newblock Sample complexity of device-independently certified quantum
  supremacy.
\newblock \emph{Phys. Rev. Lett.}, 122:\penalty0 210502, 2019.
\newblock \doi{10.1103/PhysRevLett.122.210502}.

\bibitem[Liu et~al.(2020)Liu, Arunachalam, and Temme]{liu2020rigorous}
Y.~Liu, S.~Arunachalam, and K.~Temme.
\newblock A rigorous and robust quantum speed-up in supervised machine
  learning.
\newblock \emph{arXiv preprint arXiv:2010.02174}, 2020.

\bibitem[Pietrzak(2012)]{pietrzak2012cryptography}
K.~Pietrzak.
\newblock Cryptography from learning parity with noise.
\newblock In \emph{International Conference on Current Trends in Theory and
  Practice of Computer Science}, pages 99--114. Springer, 2012.
\newblock \doi{10.1007/978-3-642-27660-6_9}.

\bibitem[Cross et~al.(2015)Cross, Smith, and Smolin]{Cross_2015}
A.~W. Cross, G.~Smith, and J.~A. Smolin.
\newblock Quantum learning robust against noise.
\newblock \emph{Phys. Rev. A}, 92\penalty0 (1):\penalty0 012327, 2015.
\newblock \doi{10.1103/physreva.92.012327}.

\bibitem[Grilo et~al.(2019)Grilo, Kerenidis, and Zijlstra]{Grilo_2019}
A.~B. Grilo, I.~Kerenidis, and T.~Zijlstra.
\newblock Learning-with-errors problem is easy with quantum samples.
\newblock \emph{Phys. Rev. A}, 99:\penalty0 032314, 2019.
\newblock \doi{10.1103/physreva.99.032314}.

\bibitem[Riste et~al.(2017)Riste, da~Silva, Ryan, Cross, C{\'o}rcoles, Smolin,
  Gambetta, Chow, and Johnson]{riste2017demonstration}
D.~Riste, M.~P. da~Silva, C.~A. Ryan, A.~W. Cross, A.~D. C{\'o}rcoles, J.~A.
  Smolin, J.~M. Gambetta, J.~M. Chow, and B.~R. Johnson.
\newblock Demonstration of quantum advantage in machine learning.
\newblock \emph{npj Quantum Inf.}, 3\penalty0 (1):\penalty0 1--5, 2017.
\newblock \doi{10.1038/s41534-017-0017-3}.

\bibitem[Kearns and Valiant(1994)]{kearns1994cryptographic}
M.~Kearns and L.~Valiant.
\newblock Cryptographic limitations on learning boolean formulae and finite
  automata.
\newblock \emph{Journal of the ACM (JACM)}, 41\penalty0 (1):\penalty0 67--95,
  1994.
\newblock \doi{10.1145/174644.174647}.

\bibitem[Servedio and Gortler(2004)]{servedio2004equivalences}
R.~A. Servedio and S.~J. Gortler.
\newblock Equivalences and separations between quantum and classical
  learnability.
\newblock \emph{SIAM J. Comp.}, 33\penalty0 (5):\penalty0 1067--1092, 2004.
\newblock \doi{10.1137/S0097539704412910}.

\bibitem[Verbeurgt(1990)]{verbeurgt1990learning}
K.~Verbeurgt.
\newblock {Learning DNF under the uniform distribution in quasi-polynomial
  time}.
\newblock In \emph{Proceedings of the third annual workshop on Computational
  learning theory}, pages 314--326, 1990.
\newblock \doi{10.1016/B978-1-55860-146-8.50027-8}.

\bibitem[Kharitonov(1993)]{kharitonov1993cryptographic}
M.~Kharitonov.
\newblock Cryptographic hardness of distribution-specific learning.
\newblock In \emph{Proceedings of the Twenty-Fifth Annual ACM Symposium on
  Theory of Computing}, STOC '93, page 372–381, New York, NY, USA, 1993.
  Association for Computing Machinery.
\newblock \doi{10.1145/167088.167197}.

\bibitem[Kharitonov(1992)]{kharitonov1995cryptographic}
M.~Kharitonov.
\newblock Cryptographic lower bounds for learnability of boolean functions on
  the uniform distribution.
\newblock In \emph{Proceedings of the Fifth Annual Workshop on Computational
  Learning Theory}, COLT '92, page 29–36, New York, NY, USA, 1992.
  Association for Computing Machinery.
\newblock \doi{10.1145/130385.130388}.

\bibitem[Arvind and Vinodchandran(1996)]{arvind1996complexity}
V.~Arvind and N.~V. Vinodchandran.
\newblock The complexity of exactly learning algebraic concepts.
\newblock In \emph{International Workshop on Algorithmic Learning Theory},
  pages 100--112. Springer, 1996.
\newblock \doi{10.1007/3-540-61863-5_38}.

\bibitem[Haussler et~al.(1991)Haussler, Kearns, Littlestone, and
  Warmuth]{haussler1991equivalence}
D.~Haussler, M.~Kearns, N.~Littlestone, and M.~K. Warmuth.
\newblock Equivalence of models for polynomial learnability.
\newblock \emph{Information and Computation}, 95\penalty0 (2):\penalty0
  129--161, 1991.
\newblock \doi{10.1016/0890-5401(91)90042-Z}.

\bibitem[Board and Pitt(1992)]{board1992necessity}
R.~Board and L.~Pitt.
\newblock On the necessity of occam algorithms.
\newblock \emph{Theoretical Computer Science}, 100\penalty0 (1):\penalty0
  157--184, 1992.
\newblock \doi{10.1145/100216.100223}.

\bibitem[Blumer et~al.(1989)Blumer, Ehrenfeucht, Haussler, and
  Warmuth]{blumer1989learnability}
A.~Blumer, A.~Ehrenfeucht, D.~Haussler, and M.~K. Warmuth.
\newblock Learnability and the vapnik-chervonenkis dimension.
\newblock \emph{Journal of the ACM (JACM)}, 36\penalty0 (4):\penalty0 929--965,
  1989.
\newblock \doi{10.1145/76359.76371}.

\bibitem[Gao et~al.(2017)Gao, Wang, and Duan]{GWD16}
X.~Gao, S.-T. Wang, and L.-M. Duan.
\newblock Quantum supremacy for simulating a translation-invariant ising spin
  model.
\newblock \emph{Phys. Rev. Lett.}, 118:\penalty0 040502, Jan 2017.
\newblock \doi{10.1103/PhysRevLett.118.040502}.

\bibitem[Bermejo-Vega et~al.(2018)Bermejo-Vega, Hangleiter, Schwarz,
  Raussendorf, and Eisert]{NewSupremacy}
J.~Bermejo-Vega, D.~Hangleiter, M.~Schwarz, R.~Raussendorf, and J.~Eisert.
\newblock Architectures for quantum simulation showing a quantum speedup.
\newblock \emph{Phys. Rev. X}, 8:\penalty0 021010, 2018.
\newblock \doi{10.1103/PhysRevX.8.021010}.

\bibitem[Miller et~al.(2017)Miller, Sanders, and Miyake]{miller_quantum_2017}
J.~Miller, S.~Sanders, and A.~Miyake.
\newblock Quantum supremacy in constant-time measurement-based computation: {A}
  unified architecture for sampling and verification.
\newblock \emph{Phys. Rev. A}, 96\penalty0 (6):\penalty0 062320, December 2017.
\newblock \doi{10.1103/PhysRevA.96.062320}.

\bibitem[Hayashi and Takeuchi(2019)]{hayashi_verifying_2019}
M.~Hayashi and Y.~Takeuchi.
\newblock Verifying commuting quantum computations via fidelity estimation of
  weighted graph states.
\newblock \emph{New J. Phys.}, 21\penalty0 (9):\penalty0 093060, 2019.
\newblock \doi{10.1088/1367-2630/ab3d88}.

\end{thebibliography}

\end{document}